\documentclass[a4paper,reqno]{amsart}

\usepackage{amsmath, amssymb, amsthm, eucal}
\usepackage{graphicx}
\usepackage[shortlabels]{enumitem}

\usepackage{xr-hyper}
\usepackage[colorlinks=true]{hyperref}
\usepackage{verbatim}
\usepackage{pdfsync}

\usepackage{comment}

\usepackage{tikz}
\usepackage{scalefnt}

\usepackage[textheight=630pt,
  textwidth=468pt,
  centering]{geometry}

\usepackage{color}

\newtheorem{theorem}{Theorem}
\newtheorem{corollary}[theorem]{Corollary}
\newtheorem{lemma}[theorem]{Lemma}
\newtheorem{proposition}[theorem]{Proposition}

\theoremstyle{definition}
\newtheorem{definition}[theorem]{Definition}
\newtheorem{remark}[theorem]{Remark}

\newcommand{\eps}{\varepsilon}

\newcommand{\abs}[1]{\left| #1 \right|}

\newcommand{\norm}[1]{\Vert #1 \Vert}

\newcommand{\RR}{\mathbb{R}}

\newcommand{\CH}{\mathcal{H}}
\newcommand{\CA}{\mathcal{A}}
\newcommand{\CC}{\mathbb{C}}
\newcommand{\PP}{\mathbb{P}}

\newcommand{\nA}{\mathcal{A}}
\newcommand{\Sph}{\mathbb{S}}

\newcommand{\brac}[1]{\left(#1\right)}
\newcommand{\fbrac}[1]{\left\{#1\right\}}
\newcommand{\sbrac}[1]{\left[#1\right]}
\newcommand{\abrac}[1]{\langle#1\rangle}

\DeclareMathOperator{\conv}{conv}

\DeclareMathOperator{\mean}{\mathbb{E}}

\DeclareMathOperator{\tr}{tr}

\DeclareMathOperator{\id}{id}

\DeclareMathOperator{\rank}{rank}

\begin{document}

\title{Stable low-rank matrix recovery via null space properties}
\date{\today}
\author{Maryia Kabanava\textsuperscript{1}, Richard Kueng\textsuperscript{2,3,4}, Holger Rauhut\textsuperscript{1}, Ulrich Terstiege\textsuperscript{1}}

\maketitle

\begin{center}
\small

\textsuperscript{1} Lehrstuhl C f\"ur Mathematik (Analysis), RWTH Aachen University, Germany \\
\textsuperscript{2} School of Physics, The University of Sydney, Australia 	\\
\textsuperscript{3} Institute for Physics \& FDM, University of Freiburg, Germany \\
\textsuperscript{4} Institute for Theoretical Physics, University of Cologne, Germany

\end{center}

\begin{abstract}
The problem of recovering a matrix of low rank from an incomplete and possibly noisy set
of linear measurements arises in a number of areas such as quantum state tomography,
machine learning and the PhaseLift approach to phaseless reconstruction problems.
In order to derive rigorous recovery results, the measurement map is usually modeled probabilistically and
convex optimization approaches including nuclear norm minimization are often used as recovery method.
In this article, we derive sufficient conditions on the minimal amount of measurements that ensure recovery via convex optimization. We establish our results via certain properties of the null space of the measurement map. 
In the setting where the measurements
are realized as Frobenius inner products with independent standard Gaussian random matrices we show
that $m > 10 r (n_1 + n_2)$ measurements are enough to uniformly and stably 
recover an $n_1 \times n_2$ matrix of rank at most $r$. Stability is meant both with respect to passing from
exactly rank-$r$ matrices to approximately rank-$r$ matrices and with respect to adding noise on the measurements.
We then significantly generalize this result by
only requiring independent mean-zero, variance one entries with four finite moments at the cost of
replacing $10$ by some universal constant. 
We also study the particular case of recovering Hermitian rank-$r$ matrices from measurement matrices proportional to rank-one projectors. 
For $r=1$, such a problem reduces to the PhaseLift approach to phaseless recovery,
while the case of higher rank is relevant for quantum state tomography.
For $m \geq C r n$ rank-one projective measurements onto independent standard Gaussian vectors, we show 
that nuclear norm minimization uniformly and stably reconstructs Hermitian rank-$r$ matrices with high probability.
Subsequently, we partially de-randomize this result by establishing an analogous statement
for projectors onto independent elements of a complex projective 4-designs at the cost of a slightly higher sampling rate $m \geq C rn \log n$.
Complex projective $t$-designs are discrete sets of vectors whose uniform distribution reproduces the first $t$ moments
of the uniform distribution on the sphere. 
Moreover, if the Hermitian matrix to be recovered is known to be positive semidefinite, then we show that the nuclear norm minimization
approach may be replaced by the simpler optimization program of minimizing the $\ell_2$-norm of the residual subject
to the positive semidefinite constraint. This has the additional advantage that no estimate of the noise level is required a 
priori. 
We discuss applications of such a result in quantum physics and the phase retrieval problem.
Apart from the case of independent Gaussian
measurements, the analysis exploits Mendelson's small ball method.

\end{abstract}


\medskip


\section{Introduction}

In recent years, the recovery of objects (signals, images, matrices, quantum states etc.) from incomplete linear measurements
has gained significant interest. While standard compressive sensing considers the reconstruction of (approximately) sparse
vectors \cite{FoucartRauhut}, we study extensions to the recovery of (approximately) low rank matrices from a small number of random
measurements. This problem arises in a number of areas such as quantum tomography \cite{gross_quantum_2010,flammia_quantum_2012,gross_focus_2013}, 
signal processing \cite{ahro15}, recommender systems \cite{care09,cata10}  and phaseless recovery \cite{candes_phase_2013,castvo13,gross_partial_2014,grkrku14}.
On the one hand, we consider both random measurement maps generated by 
independent random matrices with independent entries and on the other hand, measurements with respect to independent
rank one measurements. We derive bounds for the number of required measurements in terms of the matrix dimensions 
and the rank of the matrix that guarantee successful recovery via nuclear norm minimization. Our results are uniform and stable
with respect to noise on the measurements and with respect to passing to approximately rank-$r$ matrices. For rank-one
measurements the latter stability result is new. 

Let us formally describe our setup. 
We consider measurements of an (approximately) low-rank matrix $X\in\CC^{n_1\times n_2}$ of the form
$b = \nA(X)$, where the linear measurement map $\nA$ is given as
\begin{equation}\label{eq:MeasurementProcess}
\nA:\CC^{n_1\times n_2}\to\CC^m,\quad Z\mapsto\sum_{j=1}^m\tr(ZA_j^*)e_j.
\end{equation}
Here, $e_1,\ldots,e_m$ denote the standard basis vectors in $\CC^m$ and $A_1,\ldots,A_m\in\CC^{n_1\times n_2}$ are called measurement matrices. 
A prominent approach \cite{Fazel,RechtFazelParrilo} for recovering the matrix $X$ from $b = \nA(X)$ consists in computing
the minimizer of the convex optimization problem 
\begin{equation}\label{eqnnMinimization}
\underset{Z\in\CC^{n_1\times n_2}}\min\norm{Z}_*\quad\mbox{subject to } \nA(Z) = b,
\end{equation}
where $\norm{Z}_* = \norm{Z}_1 = \sum_{j=1}^n \sigma_j(Z)$ denotes the nuclear norm with $\sigma_j(Z)$ being
the singular values of $Z \in \CC^{n_1 \times n_2}$ and $n = \min\{n_1,n_2\}$. Efficient optimization methods exist for this
problem \cite{bopa14,bova04}. In practice the measurements are often perturbed by noise, i.e., 
\begin{equation}
\label{eq:measurements}
b=\nA(X)+w,
\end{equation}
where $w\in\CC^m$ is a vector of perturbations. In this case, we replace (\ref{eqnnMinimization}) by the noise constrained nuclear norm minimization problem
\begin{equation}\label{eqNNMinimization}
\underset{Z\in\CC^{n_1\times n_2}}\min\norm{Z}_*\quad\mbox{subject to}\;\norm{\nA(Z)-b}_{\ell_2}\leq\eta,
\end{equation}
where $\eta$ corresponds to a known estimate of the noise level, i.e., $\|w\|_{\ell_2} \leq \eta$ with $\|x\|_{\ell_p} = (\sum_{j} |x_j|^p)^{1/p}$ being the usual $\ell_p$-norm. 
In some cases it is known a priori that the matrix $X$ of interest is both Hermitian and positive semidefinite ($X \succcurlyeq 0$).
Then one may replace \eqref{eqNNMinimization} by the optimization problem
\begin{equation}\label{eq:postracemin}
\underset{Z \succcurlyeq 0 }\min \tr(Z) \quad\mbox{subject to}\;\norm{\nA(Z)-b}_{\ell_2}\leq\eta.
\end{equation}
However, as we will see, the simpler least squares problem
\begin{equation}\label{eq:posleastsquares}
\underset{Z \succcurlyeq 0 }\min \norm{\nA(Z)-b}_{\ell_2}
\end{equation}
works equally well or even better in terms of recovery under certain natural conditions. Apart from simplicity and computational efficiency it has the additional
advantage that no estimate $\eta$ of the noise level is required. 
We note that other efficient recovery methods exist as well \cite{brle10,forawa11,tawe13}, 
but we will not go into details here.

\medskip

A question of central interest concerns the minimal number $m$ of required measurements that guarantees exact (in the noiseless case) or approximate recovery.
While it is very hard to study this question for deterministic measurement maps $\nA$, several results are available for certain models of random maps.
We will study several scenarios which all have in common that the matrices $A_1,\hdots,A_m \in \RR^{n_1 \times n_2}$ in \eqref{eq:MeasurementProcess}
are independent draws of a random matrix $\Phi = (X_{ij})_{ij}$. We first consider the real-valued case, where all entries $X_{ij}$ are independent
and then move to a complex-valued scenario where $\Phi = a a^* \in \CC^{n \times n}$ is a rank one matrix generated by a random vector $a \in \CC^n$. 
For the latter scenario we consider $a$ being a complex Gaussian random vector, or $a$ being randomly drawn from a so-called (approximate) $t$-design. This last
setup has implications for quantum tomography and this part of the article can be seen as a continuation of the investigations in \cite{krt14}.
Next, we describe the present state of the art of of the various setups and present our results.

\subsection{Robust recovery from measurement matrices with independent entries} \label{sub:rank_one}

We call $\nA$ a Gaussian measurement map if the matrices $A_1,\hdots,A_m \in \RR^{n_1 \times n_2}$ in \eqref{eq:MeasurementProcess}
are independent realizations of Gaussian random matrices, i.e., all entries of the $A_j$ are independent standard Gaussian random variables.
More generally, $\nA$ is called subgaussian, if the entries of all the $A_j$ are independent, mean zero, variance one, subgaussian random variables, where we recall
that a random variable $\xi$ is called subgaussian if $\PP(|\xi| \geq t) \leq 2 e^{-c t^2}$ for some constant $c > 0$.
If 
\begin{equation}\label{bound:m:Gaussian}
m \geq C r (n_1 + n_2)
\end{equation}
for some universal constant $C>0$, then with probability at least $1- e^{-cm}$ any rank $r$ matrix $X \in \CC^{n_1 \times n_2}$ is reconstructed exactly from 
subgaussian measurements $b = \nA(X)$ via nuclear norm minimization \eqref{eqnnMinimization} \cite{RechtFazelParrilo,CandesPlan}. Moreover, 
if noisy measurements $b = \nA(X)+ w$ with $\|w \|_2 \leq \eta$ of an arbitrary matrix 
$X \in \CC^{n_1 \times n_2}$ are taken, then the minimizer $X^\sharp$ of \eqref{eqNNMinimization} satisfies, again with probability at least $1-e^{-cm}$,
\begin{equation}\label{error:estimate}
\|X - X^\sharp\|_F \leq \frac{C'}{\sqrt{r}} \inf_{Z:\rank(Z) \leq r} \|X-Z\|_* + \frac{C'' \eta}{\sqrt{m}},
\end{equation}
where $\|A\|_F = \sqrt{\tr(A^*A)}$ denotes the Frobenius norm, $\tr$ being the trace. Note that 
\[
\inf_{Z:\rank(Z) \leq r} \|X-Z\|_* = \sum_{j=r+1}^n \sigma_j(X) = \|X_c\|_*,
\]
where the singular values $\sigma_j(X)$ are arranged
in decreasing order and for $X$ with singular value decomposition $\sum_{j=1}^n \sigma_j(X) u_j v_j^*$ the matrix $X_c = \sum_{j=r+1}^n \sigma_j(X) u_j v_j^*$. 
The error estimate \eqref{error:estimate} means 
that reconstruction is robust with respect to noise on the measurements and stable with respect to passing to only approximately low rank matrices.
These statements are uniform in the sense that they hold for all matrices $X$ simultaneously once the matrix $A$ has been drawn.
They have been established in \cite{CandesPlan,MohanFazel,RechtFazelParrilo} 
via the rank restricted isometry property (rank-RIP), see e.g.~\cite{FoucartRauhut} for the standard RIP and its implications.

While the RIP is a standard tool by now, recovery of low rank matrices via nuclear norm minimization is characterized by the so-called null space property \cite{MohanFazel2,RechtXuHassibiArticle,RechtXuHassibiProceedings,FoucartRauhut,forawa11}, see below
for details. By using this concept, we are able to significantly relax from subgaussian distributions of the entries to distributions with only four finite moments.

\begin{theorem}\label{th:indepNSP}
Let $\nA:\RR^{n_1\times n_2}\to\RR^m$, $\nA(X) = \sum_{j=1}^n \tr(XA_j) e_j$, where the $A_j$ are independent copies of a random matrix $\Phi=(X_{ij})_{i,j}$ with
independent mean zero entries
obeying $\mathbb E X_{ij}^2 = 1$ and 
\[
\mathbb E X_{ij}^4\leq C_4 \quad \mbox{  for all } i,j \mbox{ and some constant } C_4.
\]
Fix $1\leq r\leq\min\{n_1,n_2\}$ and $0<\rho<1$ and set 
\[
m\geq c_1\rho^{-2}r(n_1 + n_2). 
\]
Then with probability at least $1-e^{-c_2m}$, for any $X\in\RR^{n_1\times n_2}$ the solution $X^{\sharp}$ of (\ref{eqNNMinimization}) with $b=\nA(X)+w$, $\norm{w}_{\ell_2}\leq\eta$, approximates $X$ with  error
\begin{equation}\label{eq:ErrorEstimate}
\norm{X-X^{\sharp}}_F\leq \frac{2(1+\rho)^2}{(1-\rho)\sqrt r} 
\|X_c\|_* +\frac{(3+\rho)}{(1-\rho)c_3}\cdot\frac{\eta}{\sqrt m}.
\end{equation}
Here $c_1,c_2,c_3$ are positive constants that only depend on $C_4$.
\end{theorem}
In the special case, when $\Phi$ has independent standard Gaussian entries, we apply Gordon's escape through a mesh theorem \cite{Gordon} in order to obtain an explicit constant in the estimate for the number of measurements, see Theorem \ref{th:GaussianMeas}. Roughly speaking,  with high probability, any $n_1\times n_2$ matrix of rank $r$ is stably recovered from $m>10r(n_1 + n_2)$ Gaussian measurements. 
We remark that the explicit bound $m > 3 r (n_1 + n_2)$ has been derived in \cite{chparewi10}, 
(see also \cite{mctr14} and \cite[Section 4.4]{amlomctr13} for a phase transition result in this context),
but this  bound considers nonuiform recovery, i.e.
recovery of a fixed low rank matrix with a random draw of a Gaussian measurement matrix with high probability.
Moreover, no stability under passing to approximately low rank matrices has been considered there. Our recovery result is therefore stronger than the one in \cite{chparewi10}, 
but requires more measurements.

\subsection{Robust recovery of Hermitian matrices from rank-one projective measurements} \label{sub:rank_one_measurements}

Let us now focus on the particular case of recovering complex Hermitian $n \times n$ matrices 
from noisy measurements of the form \eqref{eq:measurements}, where the measurement matrices are proportional to rank-one projectors, i.e.,
\begin{align}
A_j = a_j a_j^* \in \mathcal{H}_n	\label{eq:rank_one_measurements}
\end{align}
where $a_j \in \mathbb{C}^n$. Here, $\mathcal{H}_n$ denotes the space of complex Hermitian $n \times n$ matrices, which has real dimension $n^2$. Measurements of that type occur naturally in convex relaxations of the phase retrieval problem \cite{candes_phase_2013,castvo13,gross_partial_2014,grkrku14}. 
In fact, suppose phaseless measurements of the form $b_j = |\langle x, a_j\rangle|^2$ of a vector $x \in \mathbb{C}^n$ are given. Then we can rewrite
$b_j = \tr(x x^*a_j a_j^* ) = \tr(X A_j)$ as linear measurements of the rank one matrix $X = x x^*$. We will expand on this aspect below in Section~\ref{sub:phase_retrieval}.
Rank one measurements of low rank matrices feature prominently in quantum state tomography as well, see also below.

The prior information that the desired matrix is Hermitian limits the  
search space in the convex optimization problem \eqref{eqNNMinimization}
and it simplifies to
\begin{equation}\label{eqNNMinimizationHerm}
\underset{Z\in\mathcal H_n}\min\norm{Z}_*\quad\mbox{subject to}\;\norm{\nA(Z)-b}_{\ell_2}\leq\eta.
\end{equation}

Arguably, the most 
generic measurement matrices of the form \eqref{eq:rank_one_measurements} result from choosing each $a_j$ to be an independent complex standard Gaussian vector. 
For the particular case of phase retrieval --- i.e., where the matrix of interest $X = x x^*$ is itself proportional to a rank-one projector --- uniform recovery guarantees by means of \eqref{eqNNMinimizationHerm} have been established for $m = Cn$ independent measurements in \cite{candes_solving_2012}. 
Recently, this result has been generalized to recovery of any Hermitian  rank $r$-matrix by means of $m = Crn$ such measurements in \cite{krt14}. 
Our refined analysis of the null space property enables us to further strengthen this result by additionally guaranteeing stability under passing to  approximately low rank matrices:

\begin{theorem} \label{mainTh1}
Consider the measurement process described in (\ref{eq:MeasurementProcess}) with $m$ measurement matrices of the form \eqref{eq:rank_one_measurements},where each $a_i$ is an independent complex standard Gaussian vector.  Fix $r\leq n$, $0<\rho<1$ and suppose that
\begin{equation*}
m\geq C_1\rho^{-2}nr.
\end{equation*}
Then with probability at least $1 - \mathrm{e}^{-C_2 m}$ it holds that for any  $X \in \mathcal H_n$, 
any solution $X^\sharp$  to the convex optimization problem (\ref{eqNNMinimizationHerm}) 
with noisy measurements $b = \mathcal A (X)+\epsilon$, where $\| \epsilon \|_{\ell_2} \leq \eta$,
obeys
\begin{equation}\label{err:bound1}
\| X - X^\sharp \|_F \leq  \frac{2(1+\rho)^2}{(1-\rho)\sqrt r}\norm{X_c}_*+\frac{(3+\rho)C_3}{(1-\rho)}\cdot\frac{\eta}{\sqrt m}.
\end{equation}
 Here, $C_1,C_2$ and $C_3$ denote positive universal constants. (In particular, for $\eta=0$ and $X$ of rank at most $r$ one has exact reconstruction.)
\end{theorem}

In addition to the Gaussian measurement setting, we also consider measurement matrices that arise from taking the outer product of elements chosen independently from an approximate complex projective 4-design. 
Complex projective $t$-designs are finite sets of unit vectors in $\mathbb{C}^n$ that exhibit a very particular structure. 
Roughly speaking, sampling independently from a complex projective $t$-design, reproduces the first $t$ moments of sampling uniformly from the complex unit sphere. 
Likewise, approximate complex projective $t$-designs obey such a structural requirement approximately --- for a precise introduction, we refer to Definition~\ref{def:approx_design} below. 
As a consequence, they serve as a general purpose tool for partially de-randomizing results that initially required Gaussian random vectors \cite{kueng_spherical_2015,gross_partial_2014}. 
This is also the case here and employing complex projective  $4$-designs allows for partially de-randomizing Theorem \ref{mainTh1} at the cost of a slightly larger sampling rate.
Here, we content ourselves with presenting and shortened version of this result and refer the reader to Theorem \ref{Th2} where precise requirements on the approximate design are stated.

\begin{theorem} \label{mainTh2}
Let $r,\rho$ be as in Theorem \ref{mainTh1} and suppose that each measurement matrix $A_j$ is of the form \eqref{eq:rank_one_measurements}, where $a_j$, $j=1,\hdots,m$, are chosen independently from a (sufficiently accurate approximate) complex projective 4-design. 
If
\begin{equation*}
m \geq C_4 \rho^{-2} nr\log n,
\end{equation*}
then the assertions of Theorem \ref{mainTh1} remains valid, possibly with different universal constants. 
\end{theorem}

Note that Theorems \ref{th:indepNSP}, \ref{mainTh1}, \ref{mainTh2} resp.\ Theorem \ref{th:GaussianMeas} below  and their proofs are presented in  condensed versions in the conference papers \cite{KRTSampta1} resp. \cite{KRTSampta2}.
\subsection{Recovery of positive semidefinite matrices 
reduces to a feasibility problem} 

Imposing additional structure on the matrices to be recovered can further strengthen low rank recovery guarantees.
Positive semidefiniteness is one such structural prerequisite that, for instance, occurs naturally in the phase retrieval problem, quantum mechanics and kernel-based learning methods \cite{scholkopf_learning_2002}.
Motivated by the former, Demanet and Hand \cite{demanet_stable_2014} pointed out that  minimizing the nuclear norm --- in the sense of algorithm \eqref{eqNNMinimization} --- 
can be superfluous for recovering positive semidefinite matrices of rank one. Instead, they propose to reduce the recovery algorithm to a mere feasibility problem
and proved that such a reduction works w.h.p. for rank one projective measurements onto Gaussian vectors (the measurement scenario considered in Theorem \ref{mainTh1}). Subsequently, this recovery guarantee was strengthened by Cand\`es and Li \cite{candes_solving_2012}. 
Here, we go one step further and generalize these results to cover
uniform and stable recovery of positive semidefinite matrices of arbitrary rank. 
Relying on ideas presented in \cite{kalev_informationally_2015}, we establish the following statement.
(We refer to Section \ref{Notation} for the definition of the Schatten $p$-norm $\norm{\cdot}_p$ used in \eqref{eq:mainTh3}.)

\begin{theorem} \label{mainTh3}
Fix $r \leq n$  and 
consider the measurement processes introduced in Theorem \ref{mainTh1} (Gaussian vectors), or Theorem \ref{mainTh2} (complex projective 4-designs), respectively. 
Assume that $m\geq C_1 n r $ (in the Gaussian case) resp.\  $m\geq C_2 s nr \log n $ (in the design case), where $s\geq 1$ is arbitrary.
Then, for $1 \leq p \leq 2$ and any two positive semidefinite matrices $X,Z \in \mathcal{H}_n$,
\begin{equation}
\left\| Z - X \right\|_p \leq \frac{C_3}{r^{1-1/p}} \norm{X_c}_1 + \frac{C_4 r^{1/p-1/2}}{\sqrt{m}} \left\| \mathcal{A}(Z) - \mathcal{A}(X) \right\|_{\ell_2}
\label{eq:mainTh3}
\end{equation}
holds universally with probability
 exceeding $ 1 - \mathrm{e}^{-C_5 m}$ for the Gaussian case and $1- \mathrm{e}^{- s r}$ in the design case.
Here, $C_1,\hdots, C_5$ denote suitable positive universal constants.
\end{theorem}

This statement renders nuclear norm minimization in the sense of \eqref{eqNNMinimization} 
redundant and allows for a regularization-free estimation.
Moreover, knowledge of a noise bound $\| w \|_{\ell_2} \leq \eta$ for the measurement process \eqref{eq:measurements} is no longer required, since we can estimate any $X \succcurlyeq 0$ by solving a least squares problem of the form \eqref{eq:posleastsquares}, i.e.,
\begin{equation}
\min_{Z \in \mathcal{H}_n} \left\| \mathcal{A} (Z) - b \right\|_{\ell_2} \quad \textrm{subject to} \quad Z \succcurlyeq 0. \label{eq:least_squares}
\end{equation}
Theorem \ref{mainTh3} then in particular assures that the minimizer $Z^\sharp$ of this optimization program obeys
\begin{equation*}
\| Z^\sharp - X \|_F \leq 
\frac{C_3}{\sqrt{r}} \| X_c \|_1 + \frac{C_4}{\sqrt{m}} \left\| \mathcal{A}(Z^\sharp) - \mathcal{A}(X)  \right\|_{\ell_2}
\leq \frac{C_3}{\sqrt{r}} \| X_c \|_1 + \frac{2C_4}{\sqrt{m}}  \| w \|_{\ell_2},
\end{equation*}
where $w \in \mathbb{R}^m$ represents additive noise in the measurement process.  
It is worthwhile to mention that if a matrix $X$ of interest has rank at most $r$ and no noise is present in the sampling process \eqref{eq:measurements}, Theorem \ref{mainTh3} assures
\begin{equation}
\left\{ Z: \; Z \succcurlyeq 0, \; \mathcal{A}(Z) = \mathcal{A}(X) \right\} = \left\{ X \right\} \label{eq:feasibility_problem}
\end{equation}
with high probability. Hence, recovering $X$ from noiseless measurements indeed reduces to a feasibility problem.

We emphasize that Theorem~\ref{mainTh3} is only established for rank one projective measurements. 
For the other measurement ensembles considered here --- matrices with independent entries --- one cannot expect such a statement to hold.
This 
pessimistic prediction is due to negative results recently established in \cite[Proposition 2]{slawski_regularization_2015}.
Focusing on real matrices, the authors show that if the measurement matrices $A_j$ are chosen independently from a Gaussian orthogonal ensemble, 
then estimating any symmetric, positive semidefinite matrix $X$ via \eqref{eq:least_squares} becomes ill-posed, unless the number of measurements obeys
\begin{equation*}
 m \geq \frac{1}{4}n(n+1) = \mathcal{O}(n^2).
\end{equation*} 

Finally, we want to point out that the fruitfulness of plain least squares regression for recovering positive semidefinite matrices
was already pointed out and explored by Slawski, Li and Hein \cite{slawski_regularization_2015}.
However, there is a crucial difference in the mindset of \cite{slawski_regularization_2015} and the results presented here.
The main result \cite[Theorem 2]{slawski_regularization_2015} of Slawski et al. assumes a fixed signal $X \succcurlyeq 0$ of interest and 
provides bounds for the reconstruction error in terms of geometric properties of both $X$ and the measurement ensemble.
Conversely, Theorem \ref{mainTh3} assumes fixed measurements (e.g. $m = C rn$ projectors onto Gaussian random vectors)
and w.h.p. assures robust recovery of all matrices $X \succcurlyeq 0$ having approximately rank-$r$ simultaneously.

\subsection{Notation}{\label{Notation}} The Schatten $p$-norm of $Z\in\CC^{n_1\times n_2}$ is given by
\[
\norm{Z}_p = \brac{\sum_{j=1}^n\sigma_j(Z)^p}^{1/p},\quad p\geq 1,
\]
where $\sigma_j(Z)$, $j=1,\ldots,n$, denote the singular values of $Z$. It reduces to the nuclear norm $\norm{\cdot}_{*}$ for $p=1$ and the Frobenius norm $\norm{\cdot}_F$ for $p=2$. It is a common convention that the singular values of $Z$ are non-increasingly ordered. We write $Z=Z_r+Z_c$, where $Z_r$ is the best rank-$r$ approximation of $Z$ with 
respect to any Schatten $p$-norm of $Z$.

\section{Applications}

\subsection{Phase retrieval} \label{sub:phase_retrieval}

The problem of retrieving a complex signal $x \in \mathbb{C}^n$ from 
measurements that are ignorant towards phase information
has long been abundant in many areas of science.
Measurements of that type correspond to
\begin{equation}
b_i = \left| \langle a_i, x \rangle \right|^2 + w_i\quad i=1,\ldots,m, \label{eq:phaseless_measurements}
\end{equation}
where $a_1,\ldots,a_m \in \mathbb{C}^n$ are measurement vectors and $w_i$ denotes additive noise.  Recently, the problem's mathematical structure has received considerable attention in its own right.
It is clearly ill-posed, since all phase information is lost in the measurement process and, moreover, the measurements \eqref{eq:phaseless_measurements} are of a non-linear nature. 
This second obstacle can be overcome by a trick  \cite{balan_painless_2009} well known in conic programming: the quadratic expressions \eqref{eq:phaseless_measurements} are linear in the outer products 
$x x^*$ and $a_i a_i^*$:
\begin{equation}
b_i = \left| \langle a_i, x \rangle \right|^2 + w_i  = \tr \left( \left(a_i a_i \right)^* \left( x x^* \right) \right) + w_i. \label{eq:lifting}
\end{equation}
Note that such a ``lift'' allows for reinterpreting the phase-less sampling process as $\mathcal{A}(x x^*) = b + w$. 
Also, the new object of interest $X := x x^*$ is an Hermitian, positive semidefinite matrix of rank one. In turn, the measurement matrices $A_i = a_i a_i^*$ are constrained to be proportional to rank-one projectors. 
Consequently, such a ``lift'' turns the phase retrieval problem into a very particular instance of low rank matrix recovery --- a fact that was first observed by Cand\`es, Eldar, Strohmer and Voroninski \cite{candes_phase_2013,castvo13}. 
Subsequently, uniform recovery guarantees for $m = C n $ complex standard Gaussian measurement vectors $a_i$ have been established  which are stable towards additive noise.
The main result in \cite{candes_solving_2012} establishes with high probability that for any $X = x x^*$, solving the convex optimization problem (PhaseLift)
\begin{align}
\underset{Z \in \CH_n}{\operatorname{min}}\; \| \mathcal{A}(Z) - b \|_{\ell_1} \quad \textrm{subject to} \quad Z \succcurlyeq 0 \label{eq:PhaseLift}
\end{align}
yields an estimator $Z^\sharp$ obeying $\| Z^\sharp - x x^* \|_2 \leq C \| w \|_1 / m$.
If a bound $\| w \|_{\ell_2} \leq \eta$ on the noise in the sampling process \eqref{eq:phaseless_measurements} is available, an extension of \cite[Theorem 2]{krt14} (see section 2.3.2 in loc. cit) establishes a comparable recovery guarantee via solving
\begin{align}
\underset{Z \in \CH_n}{\operatorname{min}}\; \tr (Z) \quad \textrm{subject to} \quad \| \mathcal{A}(Z) - b \|_{\ell_2} \leq \eta, \; Z \succcurlyeq 0 \label{eq:krt_phaselift}
\end{align}
instead of PhaseLift.
Our findings allow for establishing novel recovery guarantees for retrieving phases. 
Indeed, since \eqref{eq:lifting} assures that any signal of interest is positive semidefinite and has precisely rank one, Theorem \ref{mainTh3} is applicable and yields the following corollary.

\begin{corollary} \label{cor:phaselift}
Consider $m \geq C n$ phaseless measurements of the form \eqref{eq:phaseless_measurements}, where each $a_i$ is a complex standard Gaussian vector. 
Then with probability at least $1 - \mathrm{e}^{-C' m}$ 
these measurements allow for estimating any signal $x \in \mathbb{C}^n$ 
via solving
\begin{equation}
\underset{Z \in \CH_n}{\operatorname{min}}\;  \| \mathcal{A}(Z) - b \|_{\ell_2} \quad \textrm{subject to} \quad Z \succcurlyeq 0. \label{eq:our_phaselift}
\end{equation}
The resulting minimizer $Z^\sharp$ of  \eqref{eq:our_phaselift} obeys 
\begin{equation*}
\| Z^\sharp - x x^* \|_{\ell_2} 
\leq \frac{ C \| w \|_{\ell_2}}{\sqrt{m}},
\end{equation*}
where $C$ denotes a positive constant and $w \in \mathbb{R}^m$ represents additive noise in the sampling process \eqref{eq:phaseless_measurements}.

An analogous statement is true --- with a weaker probability of success $1 - \mathrm{e}^{-s}$ for $s\geq 1$ --- for $m \geq C' s n \log (n)$ rank one projective measurements onto  independent elements of an approximate 4-design.
\end{corollary}

This recovery procedure is in spirit very similar to \eqref{eq:PhaseLift}, but it utilizes an $\ell_2$-regression instead of an $\ell_1$-norm minimization. 
Numerical studies indicate that algorithm \eqref{eq:our_phaselift} outperforms \eqref{eq:krt_phaselift} as well as \eqref{eq:PhaseLift}. 
These studies were motivated and accompany actual quantum mechanical experiments and will be published elsewhere \cite{kueng_optics_2015}.

Finally, we want to relate Corollary \ref{cor:phaselift}
to a non-convex phaseless recovery procedure  devised by Cand\`es, Li and Soltanolkotabi \cite{candes_wirtinger_2015}.
There, the authors refrain from applying the aforementioned ``lifting'' trick to render the phase retrieval problem linear. 
Instead, they use a careful initialization step, followed by  a gradient descent scheme (based on Wirtinger derivatives) to minimize 
the problem's least squares loss function directly over complex vectors $z \in \mathbb{C}^n$.
Mathematically, such an optimization is equivalent to solving
\begin{equation}
\underset{Z \in \CH_n}{\textrm{min}} \; \| \mathcal{A} (Z) - b \|_{\ell_2}
\quad \textrm{subject to} \quad  Z \succcurlyeq 0, \; \mathrm{rank}(Z) = 1	\label{eq:wirtinger}
\end{equation}
and the rank-constraint manifests the problem's non-convex nature.
Hence, the convex optimization problem \eqref{eq:our_phaselift} can be viewed as a convex relaxation of \eqref{eq:wirtinger}, obtained by omitting the non-convex rank constraint.

\subsection{Quantum information} \label{sub:quantum}

In this section we describe implications and possible applications of our findings to problems in quantum information science.
For the sake of being self-contained, we have included a brief introduction to crucial notions of quantum mechanics in the appendix. 
Quantum mechanics postulates that a finite $n$-dimensional quantum system is described by an Hermitian, positive semidefinite matrix $X$ with unit trace, called a \emph{density operator}. 
This ``quantum shape constraint'' assures that all density operators meet the requirements of Theorem \ref{mainTh3}. 
Furthermore, the rank-one projective measurements assumed in that theorem can be recast as valid quantum mechanical measurements --- see \cite[Section 3]{krt14}
for possible implementations and further discussion on this topic. 
Note, however, that such a reinterpretation is in general not possible for the measurement matrices with independent entries considered in Theorem \ref{th:indepNSP}, because these matrices fail to be Hermitian. 
With Theorem \ref{mainTh3} at hand, we underline its implications for two prominent issues in (finite dimensional) quantum mechanics.

\subsubsection{Quantum state tomography}

Inferring a quantum mechanical description of a physical system is equivalent to assigning it a \emph{density operator} (or quantum state) --- a process referred to as \emph{quantum state tomography} \cite{gross_focus_2013,ferrie_have_2015}.  Tomography is now a routine task for designing, testing and tuning qubits in the quest of building quantum information processing devices. 
Since the size of controllable quantum mechanical systems is ever increasing\footnote{Nowadays, experimentalists are able to create and control multi-partite systems of overall dimension $n = 2^8$ in their laboratories \cite{schindler_quantum_2013}. This results in a density operator of size $256 \times 256$ (a priori $65\;536$ parameters). }
it is very desirable to exploit additional structure --- if present --- when performing such a task.
One such structural property --- often encountered in actual experiments --- is \emph{approximate purity}, i.e., the density operator $X$ is well approximated by a low rank matrix. 
Performing quantum state tomography under such a prior assumption therefore constitutes a particular instance of low rank matrix recovery \cite{gross_quantum_2010,flammia_quantum_2012}.

The results presented in this paper provide recovery guarantees for tomography protocols that stably tolerate noisy measurements and moreover are \emph{robust} towards the prior assumption of approximate purity. 
In the context of tomography, results of this type so far have already been established for $m = Cnr \log^6n$ random (generalized) Pauli measurements \cite[Proposition 2.3]{liu_universal_2011} 
via proving a rank-RIP for such measurement matrices and then resorting to \cite[Lemma 3.2]{CandesPlan}.
However, this auxiliary result manifestly requires additive Gaussian noise and using a type of Dantzig, or Lasso selector to recover the best rank-$r$ approximation of a given density operator. 
This is not the case for the result established here, where performing a plain least squares regression of the form \eqref{eq:least_squares} is sufficient.

\begin{corollary} \label{cor:tomography}
Fix $r \leq n$ and suppose that the measurement operator $\mathcal{A}: \; \mathcal{H}_n \to \mathbb{R}^m$ is of the form
\begin{equation*}
\mathcal{A} (X)  = \sum_{i=1}^m \sqrt{\frac{(n+1)n}{m}} \langle a_i, X a_i \rangle e_i + w \in \mathbb{R}^m \quad \textrm{with} \quad m \geq C_1 r n \log n,
\end{equation*}
where each $a_i \in \mathbb{C}^n$ is chosen independently from an approximate 4-design and $w \in \mathbb{R}^m$ denotes additive noise.
Then, the best rank-$r$ approximation of any density operator $X$ can be obtained from such measurements
via solving
\begin{equation}
\min_{Z \in \mathcal{H}_n} \left\| \mathcal{A} (Z) - \mathcal{A} (X) \right\|_{\ell_2}
\quad \textrm{subject to} \quad Z \succcurlyeq 0, \; \tr \left( Z \right) = 1.
\label{eq:least_squares2}
\end{equation}
With probability at least $1 - \mathrm{e}^{-C_2 m}$, the minimizer $Z^\sharp$ of this optimization obeys 
\begin{equation}
\| X - Z^\sharp \|_1 
\leq C_3\| X_c \|_1 + C_4 \sqrt{r} \| w \|_{\ell_2},	\label{eq:tomography_bound}
\end{equation}
where $C_1,C_2,C_3$ and $C_4$ denote positive constants.
\end{corollary}

This statement is a direct consequence of Theorem \ref{mainTh3}. 
For the sake of clarity, we have re-scaled each projective measurement with $\sqrt{\frac{(n+1)n}{m}}$. 
This  simplifies the resulting expression \eqref{eq:tomography_bound} and moreover facilitates\footnote{In fact by resorting to the Frobenius norm bound in Theorem \ref{mainTh3} (instead of the nuclear norm bound employed to arrive at Corollary \ref{cor:tomography}),
one obtains a performance guarantee that strongly resembles \cite[Equation (8)]{liu_universal_2011} --- the main recovery guarantee in that paper.}
 direct comparison with the main result in \cite{liu_universal_2011}, as it closely mimics the scaling employed there.

Corollary \ref{cor:tomography} is valid for any type of additive noise and no a priori knowledge of its magnitude is required.
This includes the particularly relevant case of a Bernoulli error model --- see e.g. \cite[Section 2.2.2]{carpentier_uncertainty_2015} and also \cite{flammia_quantum_2012} --- which is particularly relevant for tomography experiments.
Also, note that the recovery error is bounded in nuclear norm, instead of Frobenius norm.  Such a bound is very meaningful for tomography, since quantum mechanics is a probabilistic theory and the nuclear norm encapsulates total variational distance. 
Moreover, Helstrom's theorem \cite{helstrom_quantum_1969} provides an operational interpretation of the nuclear norm distance bounded in \eqref{eq:tomography_bound}: it is proportional to the maximal bias achievable in the task of distinguishing the two quantum states $X$ and $Z^\sharp$, provided that any physical measurement can be implemented.

Finally, note that the bound on the probability of failure in Corollary \ref{cor:tomography} is much stronger than the one provided in Theorem \ref{mainTh3}. Such a strengthening is possible, because the trace of any density operator equals one. We comment on this in Remark \ref{rem:quantum_improvement} below.

\subsubsection{Distinguishing quantum states}
One crucial prerequisite in the task of inferring density operators from measurement data, is the ability to faithfully distinguish any two density operators via quantum mechanical measurements. 
The most general notion of a quantum measurement is a \emph{positive operator valued measure} (POVM) $\mathcal{M} = \left\{ E_m: E_m \succcurlyeq 0, \sum_{m} E_m = \id \right\}$ \cite[Chapter 2.2]{nielsen_quantum_2010}. 
A POVM $\mathcal{M}$ is called \emph{informationally complete} (IC)  \cite{scott_tight_2006}
if for any two density operators $X \neq Z \in \mathcal{H}_n$ there exists $E_m \in \mathcal{M} \subseteq \mathcal{H}_n$ such that
\begin{equation}
 \tr \left( E_m X \right) \neq \tr \left( E_m Z \right). \label{eq:IC}
\end{equation}
This assures the possibility of discriminating any two quantum states via such a measurement in the absence of noise. 
Without additional restrictions, such an IC POVM must contain at least $n^2$ elements. 
However, such a lower bound can be too pessimistic, if the density operators of interest have additional structure.
Approximate purity introduced in the previous subsection can serve as such an additional structural restriction:

\begin{definition}[Rank-$r$ IC, Definition 1 in \cite{heinosaari_quantum_2013}]
For $ r \leq n$, we call a POVM $\mathcal{M}=\left\{ E_m \right\}_{m \in I}$ \emph{rank-$r$ restricted informationally complete} (rank-$r$ IC), if \eqref{eq:IC} holds for any two density operators of rank at most $r$.
\end{definition}

Bounds for the number $m$ of POVM elements required to assure rank-$r$-IC have been established in \cite{heinosaari_quantum_2013,kech_role_2015,kech_quantum_2015}.
These approaches exploit topological obstructions of embeddings for establishing lower bounds and explicit POVM constructions for upper bounds.
For instance, in \cite{heinosaari_quantum_2013} a particular rank-$r$-IC POVM containing $m = 4r(n-r)-1$ elements is constructed.

Focusing less on establishing tight bounds and more on identifying entire families of rank-$r$ IC measurements,
Kalev et al. \cite{kalev_informationally_2015} observed that each measurement ensemble fulfilling the rank-RIP for some $r \leq n$ is also rank-$r$ IC. 
This in particular applies with high probability to $m = C \log^6 n \; n r$ random (generalized) Pauli measurements \cite{liu_universal_2011}. 
Theorem \ref{mainTh3}, and likewise Corollary \ref{cor:tomography}, allow us to draw similar conclusions without having to rely on any rank-RIP. 
Indeed, in the absence of noise, these results guarantee for any rank-$r$ density operator $X$
\begin{equation}
\left\{ Z: \; Z \succcurlyeq 0, \; \mathcal{A}(Z) = \mathcal{A} (X) \right\} = \left\{ X \right\} \label{eq:injectivity}
\end{equation}
with high probability. If this is the case, the measurement operator $\mathcal{A}$ allows for uniquely identifying any rank-$r$ density operator $X$. This in turn implies that $\mathcal{A}$ is rank-$r$ IC and the following corollary is immediate:

\begin{corollary}
Fix $r \leq n$ arbitrary and let $C,C'$ be absolute constants of sufficient size. Then
\begin{enumerate}
\item Any POVM containing $m = C nr$ projectors onto Haar\footnote{
Haar random vectors are vectors drawn uniformly from the complex unit sphere in $\mathbb{C}^n$. They can be obtained from complex standard Gaussian vectors by rescaling them to unit length.  Property \eqref{eq:injectivity} is invariant under such a re-scaling and Theorem \ref{mainTh1} therefore assures rank-$r$ IC for both Gaussian and Haar random vectors.
}
 random vectors is rank-$r$ IC with probability at least $1-\mathrm{e}^{C_2 m}$.
\item Any POVM containing $m = C' nr\log n$ projectors onto random elements of a (sufficiently accurate approximate) 
$4$-design is rank-$r$ IC with probability at least $1- \mathrm{e}^{-\tilde{C}_2 m}$. 
\end{enumerate}
\end{corollary}

This statement is reminiscent of a conclusion drawn in \cite{ambainis_quantum_2007,matthews_distinguishability_2009}: 
In the task of distinguishing quantum states, a POVM containing a 4-design essentially performs as good as as the uniform POVM (the union of all rank-one projectors).

\begin{remark}
In the process of finishing this article we became aware of recent work by Kech and Wolf \cite{kewo15}, who showed that the elements of a generic Parseval frame generate a rank-$r$ IC map $\mathcal{A}$ if $m \geq 4r(n-r)$. In fact, Xu 
showed in \cite{xu15} that $m \geq 4r(n-r)$ is both a sufficient and necessary condition for identifiability of complex rank $r$ matrices in $\CC^{n \times n}$. We emphasize, however, that these results are only concerned with pure identifiability and
do not come with a practical and stable recovery algorithm.  
\end{remark}

\section{The null space property for  low-rank matrix recovery}\label{sec:NSP}

Let $X\in\CC^{n_1\times n_2}$. If $X$ is only approximately of low-rank, then we would like to find a condition on the measurement map $\nA:\CC^{n_1\times n_2}\to\CC^m$ that provides the control of the recovery error by the error of its best approximation by low rank matrices. Moreover, it should also take into account that the measurements might be noisy.  
\begin{definition}
We say that $\nA:\CC^{n_1\times n_2}\to\CC^m$ satisfies the Frobenius robust rank null space property of order $r$ with constants $0<\rho<1$ and $\tau>0$ if for all $M\in\CC^{n_1\times n_2}$, the singular values of $M$ satisfy
\[
\norm{M_r}_2\leq\frac{\rho}{\sqrt r}\norm{M_c}_1+\tau\norm{\nA(M)}_{\ell_2}.
\]
\end{definition}

The stability and robustness of (\ref{eqNNMinimization}) are established by the following theorem.
\begin{theorem}\label{th:FrobeniusNSP}
Let $\nA:\CC^{n_1\times n_2}\to\CC^m$ satisfy the Frobenius robust rank null space property of order $r$ with constants $0<\rho<1$ and $\tau>0$. Let $n=\min\{n_1,n_2\}$. Then for any $X\in\CC^{n_1\times n_2}$ any solution $X^{\sharp}$ of (\ref{eqNNMinimization}) with $b=\nA(X)+w$, $\norm{w}_{\ell_2}\leq\eta$, approximates $X$ with error
\[
\norm{X-X^{\sharp}}_2\leq \frac{2(1+\rho)^2}{(1-\rho)\sqrt r}\norm{X_c}_1+\frac{2\tau(3+\rho)}{1-\rho}\eta.
\]
\end{theorem}
Theorem \ref{th:FrobeniusNSP} can be deduced from the following stronger result.
\begin{theorem}\label{th:DifferenceBetweenSignalAndFeasibleElements}
Let $1\leq p\leq 2$ and $n=\min\{n_1,n_2\}$. Suppose that $\nA:\CC^{n_1\times n_2}\to\CC^m$ satisfies the Frobenius robust rank null space property of order $r$ with constants $0<\rho<1$ and $\tau>0$. Then for any $X,Z\in\CC^{n_1\times n_2}$,
\begin{equation}\label{eq:DifferenceBetweenSignalAndFeasibleElements}
\norm{Z-X}_p\leq\frac{(1+\rho)^2}{(1-\rho) r^{1-1/p}}\brac{\norm{Z}_1-\norm{X}_1+2\norm{X_c}_1}+\frac{\tau(3+\rho)}{1-\rho}r^{1/p-1/2}\norm{\nA(Z-X)}_{\ell_2}.
\end{equation}
\end{theorem}
The proof requires some auxiliary lemmas. We start with a matrix version of Stechkin's bound.
\begin{lemma}\label{lm:2NormNuclearNorm}
Let $M\in\CC^{n_1\times n_2}$ and $r\leq\min\{n_1,n_2\}$. Then, for $p>0$,
\[
\norm{M_c}_p\leq \frac{\norm{M}_1}{r^{1-1/p}}.
\]
\end{lemma}
\begin{proof}
This follows immediately from \cite[Proposition 2.3]{FoucartRauhut}, but for convenience we give the proof.
Since the singular values of $M$ are non-increasingly ordered, it holds
\[ 
\begin{aligned}
\norm{M_c}_p^p=\sum_{j=r+1}^n(\sigma_j(M))^p&\leq(\sigma_r(M))^{p-1}\sum_{j=r+1}^n\sigma_j(M)\leq\sbrac{\frac{1}{r}\sum_{j=1}^r\sigma_j(M)}^{p-1}\sum_{j=r+1}^n\sigma_j(M)\\
&\leq\frac{1}{r^{p-1}}\norm{M}_1^{p-1}\norm{M}_1=\frac{\norm{M}_1^p}{r^{p-1}}. \qquad \qquad 
\end{aligned}
\]
\end{proof}
The next result shows that under the Frobenius robust rank null space property the distance between two matrices is controlled by the difference between their norms  and the $\ell_2$-norm of the difference between their measurements.
\begin{lemma}\label{lm:BestRApproxOfDifference}
Suppose that $\nA:\CC^{n_1\times n_2}\to\CC^m$ satisfies the Frobenius robust rank null space property of order $r$ with constants $0<\rho<1$ and $\tau>0$. Let $X,Z\in\CC^{n_1\times n_2}$ and $n=\min\{n_1,n_2\}$. Then
\[
\norm{X-Z}_1\leq  \frac{1+\rho}{1-\rho}\brac{\norm{Z}_1-\norm{X}_1+2\norm{X_c}_1}+\frac{2\tau\sqrt r}{1-\rho}\norm{\nA(X-Z)}_{\ell_2}.
\]
\end{lemma}
\begin{proof}
 Theorem 7.4.9.1 in \cite{HornJohnson} states that for matrices $A,B$ of the same size over $\CC$ 
$$
\norm{A-B}\geq \norm{\Sigma(A)-\Sigma(B)},
$$
 where $\norm{\cdot }$ is any unitarily invariant norm and $\Sigma(\cdot)$ denotes the diagonal matrix of singular values of its argument. Hence,
\[
\begin{aligned}
\norm{Z}_1&=\norm{X-(X-Z)}_1\geq\sum_{j=1}^n\abs{\sigma_j(X)-\sigma_j(X-Z)}\\
&=\sum_{j=1}^r\abs{\sigma_j(X)-\sigma_j(X-Z)}+\sum_{j=r+1}^n\abs{\sigma_j(X)-\sigma_j(X-Z)}\\
&\geq \sum_{j=1}^r\brac{\sigma_j(X)-\sigma_j(X-Z)}+\sum_{j=r+1}^n\brac{\sigma_j(X-Z)-\sigma_j(X)}.
\end{aligned}
\]
Hence, 
\[
\begin{aligned}
\norm{(X-Z)_c}_1=\sum_{j=r+1}^n\sigma_j(X-Z)&\leq\norm{Z}_1- \sum_{j=1}^r\sigma_j(X)+ \sum_{j=1}^r\sigma_j(X-Z)+ \norm{X_c}_1\\
&\leq\norm{Z}_1-\norm{X}_1+ \sqrt r\norm{(X-Z)_r}_2+2\norm{X_c}_1.
\end{aligned}
\]
Applying the  Frobenius robust null space property of $\nA$ we obtain
\[
\norm{(X-Z)_c}_1\leq\norm{Z}_1-\norm{X}_1+\rho\norm{(X-Z)_c}_1+\tau\sqrt r\norm{\nA(X-Z)}_{\ell_2}+2\norm{X_c}_1.
\]
By rearranging the terms in the above inequality we obtain 
\[
\norm{(X-Z)_c}_1\leq\frac{1}{1-\rho}\brac{\norm{Z}_1-\norm{X}_1+\tau\sqrt r\norm{\nA(X-Z)}_{\ell_2}+2\norm{X_c}_1}.
\]
In order to bound $\norm{X-Z}_1$ we use H\"older's inequality, the Frobenius robust rank null space property of $\nA$ and the inequality above,
\[
\begin{aligned}
\norm{X-Z}_1&=\norm{(X-Z)_r}_1+\norm{(X-Z)_c}_1\leq\sqrt r\norm{(X-Z)_r}_2+\norm{(X-Z)_c}_1\\
&\leq(1+\rho)\norm{(X-Z)_c}_1+\tau\sqrt r\norm{\nA(Z-X)}_{\ell_2}\\
&\leq \frac{1+\rho}{1-\rho}\brac{\norm{Z}_1-\norm{X}_1+\tau\sqrt r\norm{\nA(X-Z)}_{\ell_2}+2\norm{X_c}_1}+\tau\sqrt r\norm{\nA(X-Z)}_{\ell_2}\\
&= \frac{1+\rho}{1-\rho}\brac{\norm{Z}_1-\norm{X}_1+2\norm{X_c}_1}+\frac{2\tau\sqrt r}{1-\rho}\norm{\nA(X-Z)}_{\ell_2}.
\end{aligned}
\]
This concludes the proof. 
\end{proof}
Now we return to the proof of the theorem.
\begin{proof}[Proof of Theorem \ref{th:DifferenceBetweenSignalAndFeasibleElements}]
By H\"older's inequality, Lemma \ref{lm:2NormNuclearNorm} and the Frobenius robust rank null space property of $\nA$
\begin{align}
\norm{Z-X}_p&\leq\norm{(X-Z)_r}_p+\norm{(X-Z)_c}_p\leq r^{1/p-1/2}\norm{(X-Z)_r}_2+\norm{(X-Z)_c}_p\notag\\
&\leq\frac{\rho}{r^{1-1/p}}\norm{(X-Z)_c}_1+\tau r^{1/p-1/2}\norm{\nA(X-Z)}_{\ell_2}+\frac{1}{r^{1-1/p}}\norm{X-Z}_1\notag\\
&\leq \frac{1+\rho}{r^{1-1/p}}\norm{X-Z}_1+\tau r^{1/p-1/2}\norm{\nA(X-Z)}_{\ell_2}\label{eq:PreliminaryBound}.
\end{align}
Substituting the result of Lemma \ref{lm:BestRApproxOfDifference} into (\ref{eq:PreliminaryBound}) yields the desired inequality. 
\end{proof}

As a corollary of Theorem \ref{th:DifferenceBetweenSignalAndFeasibleElements} we obtain that if $X\in\CC^{n_1\times n_2}$ is a matrix of rank at most $r$ and the measurements are noiseless ($\eta=0$), then the Frobenius robust rank null space property implies that $X$ is the unique solution of 
\begin{equation}\label{eq:NoNoiseNNMinimization}
\underset{Z\in\CC^{n_1\times n_2}}\min\norm{Z}_1\quad\mbox{subject to}\;\nA(Z)=b.
\end{equation}
It was first stated in \cite{RechtXuHassibiProceedings} that a slightly weaker property is actually 
equivalent to the successful recovery of $X$ via (\ref{eq:NoNoiseNNMinimization}).
\begin{theorem}[Null space property]
Given $\nA:\CC^{n_1\times n_2}\to\CC^m$, every $X\in\CC^{n_1\times n_2}$ of rank at most $r$ is the unique solution of (\ref{eq:NoNoiseNNMinimization}) with $b=\nA(X)$ if and only if, for all $M\in\ker\nA\setminus\{0\}$, it holds
\begin{equation}\label{eq:NSP}
\norm{M_r}_1<\norm{M_c}_1.
\end{equation}
\end{theorem}
For the proof we refer to \cite{RechtXuHassibiProceedings} and \cite[Chapter 4.6]{FoucartRauhut}. 
According to Lemma \ref{lm:BestRApproxOfDifference}, another implication of the Frobenius robust rank null space property 
consists in the following error estimate in $\norm{\cdot}_1$ for the case of noiseless measurements,
\[
\norm{X-X^{\sharp}}_1\leq\frac{2(1+\rho)}{1-\rho}\norm{X_c}_1.
\]
The above estimate remains true, if we require that for all $M\in\ker\nA$, the singular values of $M$ satisfy
\[
\norm{M_r}_1\leq\rho\norm{M_c}_1,\quad 0<\rho<1.
\]
This property is known as the stable rank null space property of order $r$ with constant $\rho$. It is clear that if $\nA:\CC^{n_1\times n_2}\to \CC^m$ satisfies the Frobenius robust rank null space property, then it satisfies  the stable rank null space property.
The approach used in \cite{OymakMohanFazelHassibi} to verify that the stable null space property accounts for stable recovery of matrices which are not exactly of low rank, exploits the similarity between the sparse vector recovery and the low-rank matrix recovery. It shows that if some condition is sufficient for stable and robust recovery of any sparse vector with at most $r$ non-zero entries, then the extension of this condition to the matrix case is sufficient for the stable and robust recovery of any matrix up to rank $r$.

In order to check whether the measurement map $\nA:\CC^{n_1\times n_2}\to\CC^m$ satisfies the Frobenius robust rank null space property, we introduce the set
\[
T_{\rho,r}:=\left\{M\in\CC^{n_1\times n_2}:\norm{M}_2=1,\norm{M_r}_2>\frac{\rho}{\sqrt r}\norm{M_c}_1\right\}.
\] 
\begin{lemma}\label{nsplemma}
If $$\inf\{\norm{\nA(M)}_{\ell_2}:M\in T_{\rho,r}\}>\frac{1}{\tau},$$ then $\nA$ satisfies the Frobenius robust rank null space property  of order $r$ with constants $\rho$ and $\tau$.
\end{lemma}
\begin{proof} Suppose that 
\begin{equation}\label{eq:InfL2NormMeasurements}
\inf\{\norm{\nA(M)}_{\ell_2}:M\in T_{\rho,r}\}>\frac{1}{\tau}.
\end{equation}
It follows that for any $M\in\CC^{n_1\times n_2}$ such that $\norm{\nA(M)}_{\ell_2}\leq\frac{\norm{M}_2}{\tau}$ it holds
\begin{equation}\label{eq:MatrixL2StableNSP}
\norm{M_r}_2\leq\frac{\rho}{\sqrt r}\norm{M_c}_1.
\end{equation}
For the remaining $M\in\CC^{n_1\times n_2}$ with $\norm{\nA(M)}_{\ell_2}>\frac{\norm{M}_2}{\tau}$ we have
\[
\norm{M_r}_2\leq\norm{M}_2<\tau\norm{\nA(M)}_{\ell_2}.
\]
Together with (\ref{eq:MatrixL2StableNSP}) this leads to
\[
\norm{M_r}_2\leq\frac{\rho}{\sqrt r}\norm{M_c}_1+\tau\norm{\nA(M)}_{\ell_2}.
\]
for any $M\in\CC^{n_1\times n_2}$. 
\end{proof}

It is natural to expect that the recovery error gets smaller as the number of measurements increases. This can be taken into account by establishing the null space property for $\tau=\frac{\kappa}{\sqrt m}$. Then the error bound reads as follows
\[
\norm{X-X^{\sharp}}_2\leq \frac{2(1+\rho)^2}{(1-\rho)\sqrt r}\norm{X_c}_1+\frac{2\kappa(3+\rho)}{\sqrt m(1-\rho)}\eta.
\]

An important property of the set $T_{\rho,r}$ is that it is imbedded in a set with a simple structure. The next lemma relies on the ideas presented in \cite{RudelsonVershynin} for the compressed sensing setting. 
\begin{lemma}\label{lm:SetD}
Let $D$ be the set defined by
\begin{equation}\label{eqDefinitionOfD}
D:=\conv\fbrac{M\in\CC^{n_1\times n_2}:\norm{M}_2=1,\rank M\leq r},
\end{equation}
where $\conv$ stands for the convex hull.
\begin{enumerate}[(a)]
\item\label{itUnitBall} Then $D$ is the unit ball with respect to the norm
\[
\norm{M}_D:=\sum_{j=1}^L\sbrac{\sum_{i\in I_j}\brac{\sigma_i(M)}^2}^{1/2},
\]
where $L=\lceil\frac{n}{r}\rceil$, 
\[
I_j=\left\{\begin{array}{ll}
\fbrac{r(j-1)+1,\ldots,rj}, & j=1,\ldots, L-1,\\
\fbrac{r(L-1)+1,\ldots,n}, & j=L.
\end{array}\right.
\]
\item\label{itInclusion} It holds
\begin{equation}\label{eqInclusionInUniversalSet}
T_{\rho,r}\subset \sqrt{1+(1+\rho^{-1})^2}D.
\end{equation}
\end{enumerate}
\end{lemma}
Let us argue briefly why $\norm{\cdot}_D$ is a norm. Define $g:\CC^n\to[0,\infty)$ by
\[
g(x):=\sum_{j=1}^L\brac{\sum_{i\in I_j}\brac{x_i^*}^2}^{1/2},
\]
where $L$ and $I_j$ are defined in the same way as in item \ref{itUnitBall} of Lemma \ref{lm:SetD}. Then $g$ is a symmetric gauge function and $\norm{M}_D=g(\sigma(M))$ for any $M\in\CC^{n_1\times n_2}$. The norm property follows from  \cite[Theorem 7.4.7.2]{HornJohnson}.
\begin{proof}[Proof of Lemma \ref{lm:SetD}]
\ref{itUnitBall} Any $M\in D$ can be written as 
\[
M=\sum_i\alpha_iX_i
\]
with
\[
\rank X_i\leq r,\;\norm{X_i}_2=1,\; \alpha_i\geq 0,\;\sum_i\alpha_i=1.
\]
Thus
\[
\norm{M}_D\leq\sum_i\alpha_i\norm{X_i}_D=\sum_i\alpha_i\norm{X_i}_2=\sum_i\alpha_i=1.
\]
Conversely, suppose that $\norm{M}_D\leq 1$, and let $M$ have a singular value decomposition $M=U\Sigma V^*=\sum\limits_{j=1}^L\sum\limits_{i\in I_j}\sigma_i(M)u_iv_i^*$, where $u_i\in \CC^{n_1}$ and $v_i\in\CC^{n_2}$ are column vectors of $U$ and $V$ respectively. Set $M_j:=\sum\limits_{i\in I_j}\sigma_i(M)u_iv_i^*$ and $\alpha_j:=\norm{M_j}_2$, $j=1,\ldots, L$. Then each $M_j$ is a sum of $r$ rank-one matrices, so that $\rank M_j\leq r$, and we can write $M$ as 
\[
M=\sum_{j:\alpha_j\neq 0}\alpha_j\brac{\frac{1}{\alpha_j}M_j}
\]
with
\[
\sum_{j:\alpha_j\neq 0}\alpha_j=\sum_j\norm{M_j}_2=\norm{M}_D\leq 1\quad\text{and}\quad\norm{\frac{1}{\alpha_j}M_j}_2=\frac{1}{\alpha_j}\norm{M_j}_2=1.
\]
Hence $M\in D$.

\ref{itInclusion}  To prove the embedding of $T_{\rho,r}$ into a scaled version of $D$, we estimate the norm of an arbitrary element $M$ of $T_{\rho,r}$.
  According to the definition of the $\norm{\cdot}_D$-norm
\begin{align}
\norm{M}_D&=\sum_{\ell=1}^L\sbrac{\sum_{i\in I_{\ell}}\brac{\sigma_i(M)}^2}^{\frac{1}{2}}=\norm{M_r}_2+\sbrac{\sum_{i=r+1}^{2r}\brac{\sigma_i(M)}^2}^{\frac{1}{2}}+\sum_{\ell\geq 3}^L\sbrac{\sum_{i\in I_{\ell}}\brac{\sigma_i(M)}^2}^{\frac{1}{2}}\label{eq:EstimateForDNormInBlocks}.
\end{align}
To bound the last term in the inequality above, we first note that for each $i\in I_{\ell}$, $\ell\geq 3$,
\[
\sigma_i(M)\leq\frac{1}{r}\sum_{j\in I_{\ell-1}}\sigma_j(M)
\]
and hence
\[
\quad \sbrac{\sum_{i\in I_{\ell}}(\sigma_i(M))^2}^{1/2}\leq\frac{1}{\sqrt r}\sum_{j\in I_{\ell-1}}\sigma_j(M).
\]
Summing up over $\ell\geq 3$ yields
\[
\sum_{\ell\geq 3}^L\sbrac{\sum_{i\in I_{\ell}}\brac{\sigma_i(M)}^2}^{\frac{1}{2}}\leq\frac{1}{\sqrt r}\sum_{l\geq 2}\sum_{j\in I_{\ell}}\sigma_j(M)=\frac{1}{\sqrt r}\sum_{j=r+1}^n\sigma_j(M)=\frac{1}{\sqrt r}\norm{M_c}_1.
\]
and taking into account the inequality for the singular values of $M\in T_{\rho,r}$
\[
\sum_{\ell\geq 3}^L\sbrac{\sum_{i\in I_{\ell}}\brac{\sigma_i(M)}^2}^{\frac{1}{2}}\leq\rho^{-1}\norm{M_r}_2.
\]
Applying the last estimate to (\ref{eq:EstimateForDNormInBlocks}) we derive that
\[
\begin{aligned}
&\norm{M}_D\leq(1+\rho^{-1})\norm{M_r}_2\!\!+\sbrac{\sum_{i=r+1}^{2r}\brac{\sigma_i(M)}^2}^{\frac{1}{2}} \leq(1+\rho^{-1})\norm{M_r}_2+\brac{1-\norm{M_r}_2^2}^{\frac{1}{2}}.
\end{aligned}
\]
Set $a=\norm{M_r}_2$. The maximum of the function
\[
f(a):=(1+\rho^{-1})a+\sqrt{1-a^2}, \quad 0\leq a\leq 1,
\]
is attained at the point
\[
a = \frac{1+\rho^{-1}}{\sqrt{1+(1+\rho^{-1})^2}}
\]
and is equal to $\sqrt{1+(1+\rho^{-1})^2}$. Thus for any $M\in T_{\rho,r}$ it holds
\[
\norm{M}_D\leq\sqrt{1+(1+\rho^{-1})^2},
\]
which proves (\ref{eqInclusionInUniversalSet}). 
\end{proof}

\begin{remark}
The previous results hold true in the real-valued case and in the case of Hermitian 
matrices, when the nuclear norm minimization problem is solved over the set of matrices of that special type.  As a set $D$ we then take the convex hull of corresponding matrices of rank $r$ and unit Frobenius norm. The only difference in the proof of Lemma \ref{lm:SetD} occurs at the point, where we have to show that any $M$ with $\norm{M}_D\leq 1$ belongs to $D$. Say, $M\in\CC^{n_\times n}$ is Hermitian  
and $\norm{M}_D\leq 1$. Then $M=U\Lambda U^*=\sum\limits_{j=1}^L\sum\limits_{i\in I_j}\sigma_i(M)u_iu_i^*$, where $u_i\in\CC^n$, and $M_j:=\sum\limits_{i\in I_j}\sigma_i(M)u_iu_i^*$ is Hermitian.  
The rest of the proof remains unchained.    
\end{remark}

Employing the matrix representation of the measurement map $\nA$, the problem of estimating the probability of the event (\ref{eq:InfL2NormMeasurements}) is reduced to the problem of giving a lower bound for the quantities of the form $\underset{x\in T}\inf\norm{Ax}_2$. This is not an easy task for deterministic matrices, but the situation significantly changes for matrices chosen at random.

\section{Gaussian measurements}

Our main result for Gaussian measurements reads as follows.
\begin{theorem}\label{th:GaussianMeas}
Let $\nA:\RR^{n_1\times n_2}\to\RR^m$ be the linear map \eqref{eq:MeasurementProcess} generated by
a sequence $A_1, \hdots, A_m$ of independent standard Gaussian matrices, let $0<\rho<1$, $\kappa >1$ and $0<\eps<1$. If
\begin{equation}\label{eq:NumberOfMeasurementsForL2RobustRecovery}
\frac{m^2}{m+1}\geq \frac{r(1+(1+\rho^{-1})^2)\kappa^2}{(\kappa-1)^2}\left[\sqrt{n_1}+\sqrt{n_2}+\sqrt\frac{2\ln(\eps^{-1})}{r(1+(1+\rho^{-1})^2)}\right]^2,
\end{equation}
then with probability at least $1-\eps$, for every $X\in\RR^{n_1\times n_2}$, a solution $X^{\sharp}$ of (\ref{eqNNMinimization}) with $b=\nA(X)+w$, $\norm{w}_{\ell_2}\leq\eta$, approximates $X$ with  error
\[
\norm{X-X^{\sharp}}_2\leq \frac{2(1+\rho)^2}{(1-\rho)\sqrt r}\norm{X_c}_1+\frac{2\kappa\sqrt 2(3+\rho)}{\sqrt m(1-\rho)}\eta.
\]
\end{theorem}
In order to prove Theorem \ref{th:GaussianMeas} we employ Gordon's escape through a mesh theorem that provides an estimate of the probability of the event (\ref{eq:InfL2NormMeasurements}). First we recall some definitions. Let $g\in\RR^m$ be a standard Gaussian random vector, that is, a vector of independent mean zero, variance one normal distributed random variables. Then for 
\[
E_m:=\mean\norm{g}_2=\sqrt{2}\;\frac{\Gamma\brac{(m+1)/2}}{\Gamma\brac{m/2}}
\]  
we have
\[
\frac{m}{\sqrt{m+1}}\leq E_m\leq\sqrt{m},
\]
see \cite{Gordon,FoucartRauhut}. For a set $T\subset\RR^n$ we define its Gaussian width by 
\[
\ell(T):=\mean\underset{x\in T}\sup\abrac{x,g},
\]
where $g\in\RR^n$ is a standard Gaussian random vector.
\begin{theorem}[Gordon's escape through a mesh \cite{Gordon}]\label{thGordonsEscapeThroughTheMesh}
Let $A\in\RR^{m\times n}$ be a Gaussian random matrix and $T$ be a subset of the unit sphere $\Sph^{n-1}$. Then, for $t>0$, 
\begin{equation}\label{eqGordonsEscapeThroughTheMesh}
\PP\brac{\underset{x\in T}\inf\norm{Ax}_2> E_m-\ell(T)-t}\geq 1-e^{-\frac{t^2}{2}}.
\end{equation}
\end{theorem}
In order to apply this result to our measurement process (\ref{eq:MeasurementProcess}) we unravel the columns of $A_j$, $j=1,\ldots,m$, into a single row and collect all of these in a $m\times n_1n_2$-matrix $A$, so that $n = n_1 n_2$ when
applying \eqref{eqGordonsEscapeThroughTheMesh}.
In order to give a bound on the number of Gaussian measurements, Theorem \ref{thGordonsEscapeThroughTheMesh} requires to estimate the Gaussian width of the set $T_{\rho,r}$ from above. 
As it was pointed out in the previous section, $T_{\rho,r}$ is a subset of a scaled version of $D$, which has a relatively simple structure. So instead of evaluating $\ell(T_{\rho,r})$, we consider $\ell(D)$.
\begin{lemma}
For the set $D$ defined by (\ref{eqDefinitionOfD}) it holds
\begin{equation}\label{eq:EstimateOfGWOfD}
\ell(D)\leq\sqrt{r}(\sqrt{n_1}+\sqrt{n_2}).
\end{equation}
\end{lemma}
\begin{proof}
Let $\Gamma\in\RR^{n_1\times n_2}$ have independent standard normal distributed entries. Then $\ell(D)=\mean\underset{M\in D}\sup\abrac{\Gamma, M}$. Since a convex continuous real-valued function attains its maximum value at one of the extreme points, it holds $\ell(D)=\mean\underset{\substack{\norm{M}_2=1\\ \rank M\leq r}}\sup\abrac{\Gamma, M}$. By H\"older's inequality,
\[
\ell(D)\leq\mean\underset{\substack{\norm{M}_2=1\\ \rank M\leq r}}\sup\norm{\Gamma}_{\infty}\norm{M}_1\leq\sqrt r\underset{\substack{\norm{M}_2=1\\ \rank M\leq r}}\sup\norm{M}_2\mean\sigma_1(\Gamma)\leq\sqrt r(\sqrt{n_1}+\sqrt{n_2}),
\]   
where the last inequality follows from an estimate for the expectation of the largest singular value of a Gaussian matrix, see \cite[Chapter 9.3]{FoucartRauhut}. 
\end{proof}

\begin{proof}[Proof of Theorem \ref{th:GaussianMeas}]
Set $t:=\sqrt{2\ln(\eps^{-1})}$. If $m$ satisfies (\ref{eq:NumberOfMeasurementsForL2RobustRecovery}), then 
\[
E_m\brac{1-\frac{1}{\kappa}}\geq\sqrt{r(1+(1+\rho^{-1})^2)}(\sqrt{n_1}+\sqrt{n_2})+t.
\]
Together with (\ref{eqInclusionInUniversalSet}) and (\ref{eq:EstimateOfGWOfD}) this yields
\[
E_m-\ell(T_{\rho,r})-t\geq\frac{E_m}{\kappa}\geq\frac{1}{\kappa}\sqrt\frac{m}{2}.
\]
According to Theorem \ref{thGordonsEscapeThroughTheMesh}
\[
\begin{aligned}
\PP\brac{\underset{M\in T_{\rho,r}}\inf\norm{\nA(M)}_2>\frac{\sqrt m}{\kappa\sqrt 2}}\geq 1-\eps,
\end{aligned}
\]
which means that with probability at least $1-\eps$ map $\nA$ satisfies the Frobenius robust rank null space property with constants $\rho$ and $\frac{\kappa\sqrt 2}{\sqrt m}$. The error estimate follows from Theorem \ref{th:FrobeniusNSP}. 
\end{proof}

\section{Measurement matrices with independent entries and four finite moments}

In this section we prove Theorem \ref{th:indepNSP}, which is the generalization of Theorem \ref{th:GaussianMeas} to the case when the map $\nA:\RR^{n_1\times n_2}\to\RR^m$ is obtained from $m$ independent samples of a random matrix
$\Phi=(X_{ij})_{i,j}$ with the following properties:
\begin{itemize}
\item The $X_{ij}$ are independent random variables of  mean zero,
\item $\mathbb E X_{ij}^2=1$  and  $\mathbb E X_{ij}^4\leq C_4 $ for all $i,j$ and some constant $C_4$.
\end{itemize}
Note that (by H\"{o}lder's inequality) $C_4\geq 1$. 

As before the idea of the proof is to show that the event (\ref{eq:InfL2NormMeasurements}) holds with high probability. In order to do so we apply Mendelson's small ball method \cite{KoltchinskiiMendelson,Mendelson,tr14} in the manner of \cite{tr14}.
\begin{theorem}[\cite{KoltchinskiiMendelson,Mendelson,tr14}]\label{KMT}
Fix $E\subset \mathbb R^d$ and let $\phi_1,\hdots,\phi_m$ be independent copies of a random vector $\phi$ in $\mathbb R^d$. 
For $\xi >0$ let 
\[
Q_{\xi}(E;\phi) =\inf_{u\in E}\mathbb P\{ \vert \langle \phi, u \rangle \vert \geq \xi\}
\]
and
\[
W_m(E;\phi)=\mathbb E \sup_{u\in E}\langle h,u\rangle,
\]
where $h=\frac{1}{\sqrt m} \sum_{j=1}^{m}\varepsilon_j \phi_j$ with $(\varepsilon_j)$ being a Rademacher sequence \footnote{i.e., the $\varepsilon_j$ are independent and assume the values $1$ and $-1$ with probability $1/2$, respectively.}. Then for any $\xi >0$ and any $t\geq 0$ with probability at least $1-e^{-2 t^2}$ 
\[
\inf_{u\in E} \left( \sum_{i=1}^{m}\vert \langle \phi_i, u \rangle \vert^2  \right )^{1/2}\geq \xi \sqrt m Q_{2\xi}(E;\phi)-2W_m(E;\phi)-\xi t.
\]
\end{theorem}

We start with two lemmas.
\begin{lemma}\label{qabsch}
$$
\inf_{\{Y, \Vert Y \Vert_2=1 \}}\mathbb P(\vert\langle\Phi, Y \rangle \vert\geq \frac{1}{\sqrt 2} )\geq \frac{1}{4C_5}, 
$$where $C_5=\max\{3,C_4\}$.
\end{lemma}
\begin{proof}
Assume that  $Y$ has Frobenius norm one.
The  Payley-Zygmund inequality (see e.g.\ \cite[Lemma~7.16]{FoucartRauhut}, and also \cite{tr14}), implies 
\begin{equation}\label{PZ}
\mathbb P \{\vert \langle \Phi,Y\rangle \vert^2\geq \frac{1}{2}(\mathbb E \vert \langle \Phi,Y\rangle \vert^2) \}\geq \frac{1}{4}\cdot \frac{(\mathbb E \vert \langle \Phi,Y \rangle \vert^2)^2 }{\mathbb E \vert \langle \Phi,Y \rangle \vert^4 }.
\end{equation}  
We compute numerator and denominator. $$\mathbb E \vert \langle \Phi,Y\rangle \vert^2=\sum_{i,j,k,l}\mathbb E (X_{ij}X_{kl})\cdot  Y_{ij}Y_{kl}=\sum_{i,j}\mathbb E X_{ij}^2\cdot Y_{ij}^2=\sum_{i,j} Y_{ij}^2=1.$$
Likewise,
\begin{align*}
\mathbb E \vert \langle \Phi,Y\rangle \vert^4&=\sum_{i_1,\hdots,i_4, j_1,\hdots,j_4}\mathbb E (X_{i_1j_1}\cdots X_{i_4j_4})\cdot  Y_{i_1j_1}\cdots Y_{i_4j_4}\\&=\sum_{i,j}\mathbb E X_{ij}^4\cdot Y_{ij}^4+3\sum_{i_1,i_2, j_1,j_2\atop (i_1,j_1)\neq (i_2,j_2)}\mathbb E (X_{i_1j_1}^2 X_{i_2j_2}^2)\cdot  Y_{i_1j_1}^2 Y_{i_2j_2}^2\\
&=\sum_{i,j}\mathbb E X_{ij}^4\cdot Y_{ij}^4+3\sum_{i_1,i_2, j_1,j_2\atop (i_1,j_1)\neq (i_2,j_2)}   Y_{i_1j_1}^2 Y_{i_2j_2}^2
\leq \sum_{i,j}C_4\cdot Y_{ij}^4+3\sum_{i_1,i_2, j_1,j_2\atop (i_1,j_1)\neq (i_2,j_2)}   Y_{i_1j_1}^2 Y_{i_2j_2}^2\\
&\leq C_5\sum_{i_1,i_2, j_1,j_2}   Y_{i_1j_1}^2 Y_{i_2j_2}^2=C_5(\sum_{i,j} Y_{ij}^2)^2=C_5.
\end{align*}
Combining this with $(\mathbb E \vert \langle \Phi,Y\rangle \vert^2)^2=1$ and the estimate (\ref{PZ}), the claim follows. 
\end{proof}

\begin{lemma}\label{opnormabsch}
Let $\Phi_1,\hdots,\Phi_m$ be independent copies of  a random matrix $\Phi$ as above. Let $\varepsilon_1,\hdots,\varepsilon_m$ be independent Rademacher variables independent of everything else and let $H=\frac{1}{\sqrt m}\sum_{k=1}^{m}\varepsilon_k \Phi_k$. Then
$$
\mathbb E \Vert H\Vert_{\infty}\leq C_1\sqrt{n}.
$$ Here $C_1$ is a constant that only depends on $C_4$.
\end{lemma}
\begin{proof}
Let $S=\sum_{k=1}^{m} \Phi_k$.
We first desymmetrize the sum $H$ (see \cite[Lemma 6.3]{lt91}) and obtain
$$
\mathbb E\Vert H\Vert_{\infty}\leq \frac{2}{\sqrt{m}}\mathbb E\Vert S\Vert_{\infty}.
$$

Therefore, it is enough to show that $\mathbb E\Vert S\Vert_{\infty}\leq c_3\sqrt{mn}$ for a suitable constant $c_3.$
The matrix $S$ has independent mean zero entries, hence by a result Lata{\l}a (see \cite{la05})  the following estimate holds for some universal constant $C_2$,
$$
\mathbb E \Vert S\Vert_{\infty}\leq C_2\left(\max_{i} \sqrt{\sum_j \mathbb E S_{ij}^2}+\max_{j} \sqrt{\sum_i \mathbb E S_{ij}^2}+\sqrt[4]{\sum_{i,j} \mathbb E S_{ij}^4}\right).
$$
 Denoting the entries of $\Phi_k$ by $X_{k;ij}$, we have $S_{ij}=\sum_k X_{k;ij}$. Hence, using the independence of the $X_{k;ij}$, we obtain   $\mathbb E S_{ij}^2=\mathbb E (\sum_k X_{k;ij})^2=\sum_k \mathbb E X_{k;ij}^2=m$. Thus, $\sqrt{\sum_j \mathbb E S_{ij}^2} \leq \sqrt{nm}$ for any $i$ and $ \sqrt{\sum_i \mathbb E S_{ij}^2}\leq \sqrt{nm}$ for any $j$.
Finally to estimate $\sqrt[4]{\sum_{i,j} \mathbb E S_{ij}^4}$ we calculate $\mathbb E S_{ij}^4=\mathbb E (\sum_k{X_{k;ij}})^4$. Using again that  the $X_{k;ij}$ are independent and have mean zero we obtain $$
 E S_{ij}^4=\sum_k \mathbb E X_{k;ij}^4 + 3\sum_{k_1\neq k_2}\mathbb E X_{k_1;ij}^2\mathbb E X_{k_2;ij}^2.
$$ Using that $\mathbb E X_{k;ij}^2=1$ for all $i,j,k$, we obtain $\mathbb E S_{ij}^4\leq C_5m^2$, where $C_5=\max\{3,C_4\}$ and hence   $$
\sqrt[4]{\sum_{i,j} \mathbb E S_{ij}^4}\leq \sqrt[4]{C_5m^2n^2}=\sqrt[4]{C_4}\sqrt{mn}.$$ Hence, indeed $\mathbb E\Vert S\Vert_{\infty}\leq c_3\sqrt{mn}$ for a suitable constant $c_3$ that depends only on $C_4$. 
\end{proof}

\begin{proof}[Proof Theorem \ref{th:indepNSP}]
Let now $T_{\rho,r}$ and $D$ be the sets defined in Section \ref{sec:NSP}, but restricted to the real-valued matrices.
By H\"older's inequality, for any  $n_1\times n_2$ matrix $Y$ of Frobenius norm $1$ and rank at most $r$ and any $n_1\times n_2$ matrix $H$,
$$
 \langle H, Y \rangle \leq \Vert Y \Vert_1 \Vert H \Vert_{\infty}\leq \sqrt r \Vert H \Vert_{\infty}.
$$
Hence
\begin{equation}\label{wmabsch}
\sup_{Y\in D} \langle H, Y \rangle \leq  \sqrt r \Vert H \Vert_{\infty}.
\end{equation}

\noindent  Let $H=\frac{1}{\sqrt m} \sum_{j=1}^{m}\varepsilon_j \Phi_j$ and
let $\xi = \frac{1}{2\sqrt 2}$ and $E=T_{\rho,r}$. Then it follows from Theorem \ref{KMT} that for any $t\geq 0$ with probability at least $1-e^{-2 t^2}$ 
\begin{equation}\label{KMTAnw}
\inf_{Y\in T_{\rho,r}} \left( \sum_{i=1}^{m}\vert \langle \Phi_i, Y \rangle \vert^2  \right )^{1/2}\geq   \frac{\sqrt m}{2\sqrt 2}Q_{\frac{1}{\sqrt 2}}(T_{\rho,r};\Phi)-2 W_m(T_{\rho,r},\Phi)  -\frac{1}{2\sqrt 2} t.
\end{equation}
Using Lemma \ref{qabsch} and the fact that all elements of $T_{\rho,r}$ have Frobenius norm $1$, we obtain 
\begin{equation}\label{Q1/2}
Q_{\frac{1}{\sqrt 2}}(T_{\rho,r};\Phi)\geq \frac{1}{4C_5}.
\end{equation}
Combining now the fact that $T_{\rho,r}\subseteq \sqrt{1+(1+\rho^{-1})^2}D $ (see Lemma \ref{lm:SetD}) with  estimate (\ref{wmabsch}) and  Lemma \ref{opnormabsch} leads to
\begin{equation}\label{WmT}
W_m(T_{\rho,r},\Phi)\leq \sqrt{1+(1+\rho^{-1})^2} \sqrt r\ \mathbb E \norm{H}_{\infty} \leq C_1 \sqrt{1+(1+\rho^{-1})^2} \sqrt r \sqrt n.
\end{equation}
Using (\ref{KMTAnw}), (\ref{Q1/2}) and (\ref{WmT}) we see that 
choosing $m\geq c_1\rho^{-2}nr$  and $t=c_4m$ for suitable constants $c_1,c_4$, we obtain with probability at least $1-e^{-c_2m}$
$$ \inf_{Y\in T_{\rho,r}} \left( \sum_{i=1}^{m}\vert \langle \Phi_i, Y \rangle \vert^2  \right )^{1/2}\geq c_3 \sqrt m$$ for suitable constants $c_2,c_3$. Now the claim follows from Lemma \ref{nsplemma} and Theorem \ref{th:FrobeniusNSP} (both of which also hold in the real valued version by the same proofs respectively). 
\end{proof}

\section{Rank one Gaussian measurements}

In this section we prove Theorem \ref{mainTh1}.  
The proof technique is an application of Mendelson's small ball method analogous to the proof of Theorem \ref{th:indepNSP}.
Let
\[
T^{\mathcal H}_{\rho,r}:=\fbrac{M\in\mathcal H_n: \norm{M}_2=1,\ \norm{M_r}_2>\frac{\rho}{\sqrt r}\norm{M_c}_1}.
\] 
Let $T_{\rho,r}$ be defined as $T^{\mathcal H}_{\rho,r}$ but with $\mathcal H_n$ replaced by the set of all complex $n\times n$-matrices (i.e. it is defined as before with $n_1=n_2=n$).
Then $T^{\mathcal H}_{\rho,r}\subseteq T_{\rho, r}$.  It is enough to show that with high probabiliy
\begin{equation}
 \inf_{Y\in T_{\rho, r}^{\mathcal H}} \left( \sum_{j=1}^{m}\vert \langle a_ja_j^*, Y \rangle \vert^2  \right )^{1/2}\geq \sqrt m/C_3 \label{eq:Gauss_tau}
\end{equation}
We apply Theorem~\ref{KMT} with $E=T_{\rho, r}^{\mathcal H}$.
The next lemma estimates the small ball probability $Q_{\frac{1}{\sqrt 2}}(E;\phi) $ used in Mendelson's method.
\begin{lemma}[see \cite{krt14}]
$Q_{\frac{1}{\sqrt 2}}(E;\phi)  :=\inf_{u\in E}\mathbb P\{ \vert \langle aa^*, u \rangle \vert \geq \frac{1}{\sqrt 2}\} \geq \frac{1}{96}$.
\end{lemma}

Let now (as in \cite{tr14,krt14})
\begin{equation}
H=\frac{1}{\sqrt m}\sum_{j=1}^m\varepsilon_ja_ja_j^*, \label{eq:H}
\end{equation}
where the $\varepsilon_j$  form a Rademacher sequence.
For any  $M\in \mathcal H_n$ and any $n\times n$ matrix $Y$ of Frobenius norm $1$ and rank at most $r$ $$
 \langle M, Y \rangle \leq \Vert Y \Vert_1 \Vert M \Vert_{\infty}\leq \sqrt r \Vert M \Vert_{\infty}.
$$
Since $E=T_{\rho, r}^{\mathcal H} \subseteq T_{\rho, r}\subseteq \sqrt{1+(1+\rho^{-1})^2}D$, this implies $$W_m(E,\phi) =\mathbb E \sup_{Y\in E}\langle H,Y\rangle \leq \sqrt{1+(1+\rho^{-1})^2}\sqrt r \mathbb E \Vert H\Vert_{\infty}.$$
As in \cite{krt14} we use now that by the arguments in \cite[Section~5.4.1]{ve12} we have $\mathbb E \Vert H\Vert_{\infty}\leq c_2 \sqrt n$ if $m\geq c_3 n$ for suitable constants $c_2, c_3$, 
see also \cite[Section~8]{tr14}. Now the claim of Theorem \ref{mainTh1} follows from   Theorem \ref{KMT}, comp. the proof of Theorem  \ref{th:indepNSP}.
\qed

\begin{remark}
Inspecting the above proof, resp.\ the proofs of the cited statements in \cite{krt14}, we see that the real valued analogue of Theorem \ref{mainTh1} is also true. We even may assume for this  that the $a_j$ are i.i.d.\ subgaussian  with $k$-th moments, where $k\leq 8$, equal to the corresponding $k$-th moments of the Gaussian standard distribution. The constants then depend only on the distribution of the $a_j$. We also note that a similar statement in the  real case  for the recovery of positive semidefinite matrices using subgaussian measurements has  been shown by Chen, Chi and Goldsmith in \cite{cchg13} using the rank restricted isometry property.
\end{remark}

\section{Rank one measurements generated by 4-designs}

Recall the definition of an approximate, weighted $t$-design. 

\begin{definition}[\emph{Approximate $t$-design}, Definition 2 in \cite{ambainis_quantum_2007}] \label{def:approx_design}
We call a weighted set  $\left\{p_i,w_i \right\}_{i=1}^N$ of normalized vectors
an approximate $t$-design of $p$-norm accuracy $\theta_p$, if
\begin{equation}
 \left\| \sum_{i=1}^N p_i \left( w_i w_i^* \right)^{\otimes t} - \int_{\|w \|_{\ell_2} = 1} \left( w w^* \right)^{\otimes t} \mathrm{d}w \right\|_p \leq  \binom{n+t-1}{t}^{-1} \theta_p.
\label{eq:approx_designs}
\end{equation}
\end{definition}

A set of unit vectors obeying $\theta_p = 0 $ for  $1 \leq p \leq \infty$ is called an \emph{exact $t$-design}, 
 see \cite{scott_tight_2006} and also \cite{krt14,gross_partial_2014}.

\begin{theorem} \label{Th2}
Let $\left\{p_i,w_i \right\}_{i=1}^N$ be a an approximate $4$-design with either $\theta_\infty \leq 1/(16r^2)$, or $\theta_1 \leq 1/4$ that furthermore obeys
$
\left\| \sum_{i=1}^N p_i w_i w_i^* - \frac{1}{n} \id \right\|_\infty \leq \frac{1}{n}
$.
Suppose that the measurement operator $\mathcal{A}$ is generated by 
\begin{equation*}
m\geq C_4\rho^{-2}nr\log n
\end{equation*}
measurement matrices $A_j =  \sqrt{n(n+1)}a_j a_j^*$, where each $a_j$ is drawn independently from $\left\{p_i, w_i \right\}_{i=1}^N$. 
Then, with probability at least $1 - \mathrm{e}^{-C_5 m}$, 
$\mathcal{A}$ obeys the Frobenius robust rank null space property of order $r$ with constants $0 < \rho < 1$ and $ \tau = C_6/\sqrt{m}$. 
Here, $C_4,C_5$ and $C_6$ denote  positive constants depending only on the design.
\end{theorem}

Theorem \ref{mainTh2} readily follows from combining this statement with Theorem \ref{th:DifferenceBetweenSignalAndFeasibleElements}.

\begin{proof}[Proof of Theorem \ref{Th2}]
We start by presenting a proof for measurements drawn from an exact 4-design.
Paralleling the proof of Theorem \ref{mainTh1}, the statement can be deduced from Theorem \ref{KMT} by utilizing results from \cite{krt14}.
Provided that $a$ is randomly chosen from a re-scaled, weighted $4$-design (such that each element has Euclidean length $\| w_i \|_{\ell_2} = \sqrt[4]{(n+1)n}$), 
\cite[Proposition 12]{krt14} implies that
\begin{equation}
\inf_{Z \in T_{\rho,r}} \mathbb{P} \left( |  \mathrm{tr} \left( a a^* Z \right)| \geq \xi \right) \geq  \inf_{\| Z \|_2 =1} \mathbb P \left( |  \mathrm{tr} \left( a a^* Z \right)| \geq \xi \right) \geq \frac{(1-\xi^2)^2}{24} \label{eq:designs_Q}
\end{equation}
is valid for all $\xi \in [0,1]$.
Now let $H = \sum_{i=1}^m \epsilon_i a_i a_i^*$ be as in Theorem \ref{KMT}. Lemma \ref{lm:SetD} together with the fact that $D$ is the convex hull of all matrices of rank at most $r$ and Frobenius norm 1
allows us to conclude for $m \geq 2 n \log n$, that,
\begin{align*}
W_m \left( T_{\rho,r}, aa^* \right)
&= \mathbb{E}  \sup_{M \in T_{\rho,r}} \tr \left( H M \right) 
\leq \sqrt{1+ (1+\rho^{-1})^2} \;\mathbb{E} \sup_{M \in D} \tr \left( H M \right) \\
&\leq \sqrt{1+ (1+\rho^{-1})^2} \sup_{M \in D} \| M \|_1  \mathbb{E} \| H \|_\infty 
\leq   \sqrt{1+ (1+\rho^{-1})^2} \sqrt{r} \;  \mathbb{E} \| H \|_\infty \\
& \leq 3.1049  \sqrt{1+ (1+\rho^{-1})^2 r n \log (2n)},
\end{align*}
where the last bound is due to  \cite[Proposition 13]{krt14}. 
Fixing $0 < \xi < 1/2$ arbitrarily and inserting these two bounds into Theorem~\ref{KMT} completes the proof.

An analogous statement for approximate 4-designs --- with slightly worse absolute constants --- can be obtained by resorting to
the generalized versions of \cite[Propositions 12 and 13]{krt14} presented in Section~4.5.1 in loc.\ cit.\
which are valid for approximate 4-designs that satisfy the conditions stated in Theorem~\ref{Th2}. 
\end{proof}

\section{The positive semidefinite case}

Finally, we focus on the case, where the matrices of interest are Hermitian and positive semidefinite and establish Theorem \ref{mainTh3}. 
In order to arrive at such a statement, we closely follow the ideas presented in \cite{kalev_informationally_2015} 
which in turn were inspired by \cite{bruckstein_uniqueness_2008} containing an analogous statement for a non-negative compressed sensing scenario. 

We require two further concepts from matrix analysis.
For every positive semidefinite matrix $W \succcurlyeq 0$ with eigenvalue decomposition $W = \sum_{i=1}^n \lambda_i w_i w_i^*$ we define its square root to be
$W^{1/2} := \sum_{i=1}^n \sqrt{\lambda_i} w_i w_i^*$. In other words, $W^{1/2}$ is the unique positive semidefinite 
matrix which acts on the eigenspace corresponding to the eigenvalue $\lambda_i$ of $W$ by 
multiplication by $\sqrt{\lambda_i}$. Note that this matrix  obeys $W^{1/2}\cdot W^{1/2} = W$. 
Also, recall that the condition number $\kappa (W)$ of a matrix $W$ is the ratio between its largest and smallest nonzero singular value. For an invertible Hermitian matrix with inverse $W^{-1}$ this number equals
\begin{equation*}
\kappa (W) = \| W \|_\infty \| W^{-1} \|_\infty.
\end{equation*}

Suppose that the measurement process \eqref{eq:measurements} is such that there exists $t\in\RR^m$ which assures that $W:=\sum_{j=1}^m t_jA_j$ is positive definite. We define the artificial measurement map 
\begin{equation}
\nA_{W^{1/2}}:\CH_n\to\RR^m, \quad Z\mapsto \nA(W^{-1/2}Z W^{-1/2}) \label{eq:artificial_measurements}
\end{equation}
and the endomorphism
\begin{equation}
Z \mapsto \tilde{Z} := W^{1/2} Z W^{1/2} \label{eq:mapping}
\end{equation}
of $\mathcal{H}_n$. 
Note that these definitions assure
\begin{equation}\label{eq:RelationBetweenAAB}
\nA(Z)=\nA_{W^{1/2}}(\tilde Z) \quad \mbox{ for all } Z \in \mathcal{H}_n
\end{equation}
and the singular values of $Z$ and $\tilde Z$ satisfy
\begin{equation}\label{eq:SingularValuesZandTildeZ}
 \sigma_j(\tilde Z)\leq \norm{W^{1/2}}_{\infty}^2\sigma_j(Z)=\norm{W}_{\infty}\sigma_j(Z),\quad \sigma_j(Z)\leq \norm{W^{-1/2}}_{\infty}^2\sigma_j(\tilde Z)=\norm{W^{-1}}_{\infty}\sigma_j(\tilde Z),
 \end{equation}
see \cite[p. 75]{BhatiaMatrixAnalysis}. Consequently, the mapping \eqref{eq:mapping} preserves the rank of any matrix.
The following result assures that the artificial measurement operator $\mathcal{A}_{W^{1/2}}$ obeys the Frobenius robust rank null space property, if the original $\mathcal{A}$ does. 

\begin{lemma}\label{NSPvgl}
Suppose that $\CA$ satifies the Frobenius robust rank null space property of order $r$ with constants $\rho$ and $\tau$ and suppose that $W=\sum_{j=1}^m t_jA_j$ is positive definite.
Then $\CA_{W^{1/2}}$ also  obeys the  Frobenius robust rank null space property of order $r$, but with constants $\tilde \rho= \kappa (W) \rho$ and $\tilde \tau=\norm{W}_{\infty}\tau$.
\end{lemma}
\begin{proof}
Let $\tilde Z\in\CH_n$. 
Relations (\ref{eq:RelationBetweenAAB}), (\ref{eq:SingularValuesZandTildeZ}) together with the Frobenius robust rank null space property of $\nA$ imply that
\[
\begin{aligned}
\norm{\tilde Z_r}_2&\leq\norm{W^{1/2}}_{\infty}^2\norm{Z_r}_2\leq\norm{W}_{\infty}\brac{\frac{\rho}{\sqrt r}\norm{Z_c}_1+\tau\norm{\nA(Z)}_{\ell_2}}\\
&\leq \frac{\norm{W}_{\infty}\norm{W^{-1}}_{\infty}\rho}{\sqrt r}\norm{\tilde Z_c}_1+\norm{W}_{\infty}\tau\norm{\nA_{W^{1/2}}(\tilde Z)}_{\ell_2}. \qquad 
\end{aligned}
\] \end{proof}

\begin{lemma}\label{Normdiff}
Suppose there is $t\in\RR^m$ such that $W:=\sum_{j=1}^m t_jA_j$ is positive definite. Let $\tilde X,\tilde Z$ be positive semidefinite. Then,
$$
\norm{\tilde{Z}}_1-\norm{\tilde{X}}_1\leq \norm{t}_{\ell_2}\norm{\nA_{W^{1/2}}(\tilde Z-\tilde X)}_{\ell_2}.
$$
\end{lemma}

\begin{proof}
The claim follows from positive semidefiniteness of both $\tilde{Z}$ and $\tilde{X}$ and  our choice of the endomorphism (\ref{eq:mapping}).
Indeed,
\begin{align*}
\norm{\tilde Z}_1&=\tr(\tilde Z- \tilde X)+\norm{\tilde X}_1=\tr(W^{1/2}(Z-X)W^{1/2})+\norm{\tilde X}_1=\tr(W(Z-X))+\norm{\tilde X}_1\\
&=\sum_{j=1}^m t_j\tr(A_j(Z-X))+\norm{\tilde X}_1=\langle t,\CA(Z-X) \rangle +\norm{\tilde X}_1\\
&=\abrac{t,\nA_{W^{1/2}}(\tilde Z-\tilde X)}+\norm{\tilde X}_1\leq\norm{t}_{\ell_2}\norm{\nA_{W^{1/2}}(\tilde Z-\tilde X)}_{\ell_2}+\norm{\tilde X}_1.
\end{align*}
Here $X$ resp.\ $Z$ denote the preimage of $\tilde X$ resp $\tilde Z$ under the map (\ref{eq:mapping}). 
\end{proof}

This simple technical statement allows us to establish the main result of this section.

\begin{theorem}\label{th:NSPforPositiveMatrices}
Suppose there exists $t\in\RR^m$ such that $W:=\sum_{j=1}^m t_jA_j$ is positive definite and $\nA$ satisfies the Frobenius robust rank null space property with constants $0<\rho<\frac{1}{\kappa(W)}$ and $\tau>0$. Let $1\leq p\leq 2$. Then, for any $X,Z\succcurlyeq 0$,
\begin{equation}
\norm{Z-X}_p\leq\frac{2C\kappa(W)}{r^{1-1/p}}\norm{X_c}_1+r^{1/p-1/2}\norm{\nA(Z)-\nA(X)}_{\ell_2} \norm{W^{-1}}_{\infty} \brac{\frac{C\norm{t}_2}{\sqrt r}+D\norm{W}_{\infty}\tau}
\label{eq:NSPforPositiveMatrices}
\end{equation}
with constants
$
C=\frac{(1+\kappa(W)\rho)^2}{1-\kappa(W)\rho}
$
and
$ D=\frac{3+\kappa(W)\rho}{1-\kappa(W)\rho}.
$
\end{theorem}

\begin{proof}
Let $X,Z \succcurlyeq 0 $ be arbitrary. Then
\begin{equation*}
\| Z - X \|_p = \left\| W^{-1/2} \left( \tilde{Z} - \tilde{X} \right) W^{-1/2} \right\|_p \leq \| W^{-1} \|_\infty \| \tilde{Z} - \tilde{X} \|_p
\end{equation*}
holds and the resulting matrices $\tilde{Z},\tilde{X}$ are again positive-semidefinite. 
Also, since $\nA$ satisfies the Frobenius robust rank null space property with constants $0<\rho<\frac{1}{\kappa(W)}$ and $\tau>0$, 
Lemma \ref{NSPvgl} assures that $\mathcal{A}_{W^{1/2}}$ does the same with constants $0 < \tilde{\rho} < 1$ and $\tilde{\tau} = \| W \|_\infty \tau >0$.
Combining this with Theorem \ref{th:DifferenceBetweenSignalAndFeasibleElements} and Lemma \ref{Normdiff} implies
\[
\begin{aligned}
\norm{\tilde Z-\tilde X}_p& \leq\frac{C}{r^{1-1/p}}\brac{\norm{\tilde Z}_1-\norm{\tilde X}_1+2\norm{\tilde X_c}_1}+D\norm{W}_{\infty}\tau r^{1/p-1/2}\norm{\nA_{W^{1/2}}(\tilde Z-\tilde X)}_{\ell_2}\\
&\leq\frac{C}{r^{1-1/p}}\brac{\norm{t}_{\ell_2}\norm{\nA_{W^{1/2}}(\tilde Z-\tilde X)}_{\ell_2}+2\norm{\tilde X_c}_1}+D\norm{W}_{\infty}\tau r^{1/p-1/2}\norm{\nA_{W^{1/2}}(\tilde Z-\tilde X)}_{\ell_2}\\
&\leq \frac{2C}{r^{1-1/p}}\norm{\tilde X_c}_1+r^{1/p-1/2}\norm{\nA_{W^{1/2}}(\tilde Z-\tilde X)}_{\ell_2}\brac{\frac{C\norm{t}_{\ell_2}}{\sqrt r}+D\norm{W}_{\infty}\tau }.
\end{aligned}
\]
The desired statement follows from this estimate by taking into account (\ref{eq:RelationBetweenAAB}) and (\ref{eq:SingularValuesZandTildeZ}). 
\end{proof}

Note that in contrast to other recovery guarantees established here, Theorem~\ref{th:NSPforPositiveMatrices} does not require any convex optimization procedure. 
However, it does require the measurement process to obey an additional criterion: the intersection of the span of measurement matrices with the cone of positive definite matrices must be non-empty.
We show that this is the case for the rank-one projective measurements introduced in the previous section with high probability. 
Since it has already been established that sufficiently many measurements of this kind obey the Frobenius robust rank null space property with high probability (see Theorems \ref{mainTh1} and \ref{Th2} and their respective proofs),
Theorem \ref{mainTh3} can then be established by taking the union bound over the individual probabilities of failure.

\begin{proposition} \label{prop:WGauss}
Suppose $m \geq 4 n$  and let $A_1,\ldots,A_m$ be matrices of the form $a_j a_j^*$, where each $a_i \in \mathbb{C}^n$ is a random complex standard Gaussian vector.
Then with probability at least $1- 2\mathrm{e}^{-C_{10} m}$, $W := \frac{1}{m} \sum_{j=1}^m A_j$ is positive definite and obeys
\begin{equation}
\max \left\{ \| W \|_\infty, \| W^{-1} \|_\infty,  \kappa (W) \right\}  \leq C_{11}.
\label{eq:WGauss}
\end{equation}
Here, $C_9, C_{10}, C_{11}>0$ denote universal positive constants.
\end{proposition}

Note that such a construction corresponds to setting $t = \frac{1}{m} (1,\ldots,1)^T \in \mathbb{R}^m$ which obeys $\| t \|_{\ell_2} = 1/\sqrt{m}$.

\begin{proof}
For the sake of simplicity, we are going to establish the statement for real standard Gaussian vectors. Establishing the complex case can be done analogously and leads to slightly different constants.
Let $e_1,\ldots,e_m$ denote the standard basis in $\mathbb{R}^m$. We define the auxiliary $m \times n$ matrix $A := \sum_{i=1}^m e_i a_i^*$ which obeys
\begin{equation*}
\frac{1}{m} A^T A = \frac{1}{m} \sum_{i=1}^m a_i e_i^* \sum_{j=1}^m e_j a_j^* = \frac{1}{m} \sum_{i=1}^m a_i a_i^* = \frac{1}{m} \sum_{i=1}^m A_i = W. \label{eq:wishart}
\end{equation*}
Also, by construction, $A$ is a random matrix with standard Gaussian entries.
Essentially, this relation implies that $m W$ is Wishart-distributed.
From \eqref{eq:wishart} and the defining properties of eigen- and singular values we infer that
\begin{equation}
\sqrt{ \lambda_{\min} (W)} = \frac{1}{\sqrt{m}} \sqrt{ \lambda_{\min} \left( A^T A \right) } = \frac{1}{\sqrt{m}} \lambda_{\min} \left( \sqrt{A^T A} \right) = \frac{1}{\sqrt{m}} \sigma_{\min} (A)
\label{eq:eig_sing}
\end{equation}
and an analogous statement is true for the largest eigenvalue $\lambda_{\max}(W)$.
Since $A$ is a Gaussian $m \times n$ matrix, concentration of measure implies that for any $\tilde{\tau} > 0$
\begin{equation}
\sqrt{m} - \sqrt{n} -\tilde{\tau} \leq \sigma_{\min}(A) \leq \sigma_{\max}(A) \leq \sqrt{m} + \sqrt{n} + \tilde{\tau} \label{eq:wishart_aux1}
\end{equation}
with probability at least $1 - 2\mathrm{e}^{-\tilde{\tau}^2/2}$ --- see e.g. \cite[Corollary 5.35]{ve12} or \cite[Theorem 9.26]{FoucartRauhut}. 
Combining this with \eqref{eq:eig_sing}, recalling the assumption $m \geq 4 n$ and defining $\tau = \tilde{\tau}/\sqrt{m}$ allows for establishing
\begin{equation*}
\frac{1}{2} - \tau \leq 1 - \sqrt{\frac{n}{m}} - \tau \leq \sqrt{\lambda_{\min}(W)} \leq \sqrt{\lambda_{\max} (W)} \leq 1 + \sqrt{\frac{n}{m}} + \tau \leq \frac{3}{2} + \tau
\end{equation*}
with probability at least $1 - 2\mathrm{e}^{-m\tau^2/2}$. This inequality chain remains valid, if we square the individual terms. Setting $\tau = 1/4$ thus allows us to conclude
\begin{equation}
\max \left\{ \lambda_{\max}(W), \lambda_{\min}^{-1} (W),  \frac{\lambda_{\max}(W)}{\lambda_{\min}(W)} \right\} 
\leq \left( \frac{ 3/2 + \tau}{1/2 - \tau} \right)^2 = 49 = C_{11},
\end{equation}
with probability at least $1 - 2\mathrm{e}^{-m/32}$. 
\end{proof}

Alternatively, we could have relied on bounds on the condition number of Gaussian random matrices presented in \cite{chen_condition_2005}. 
While these bounds would be slightly tighter, we feel that our derivation is more illustrative and it suffices for our purpose.

\begin{proposition} \label{prop:W4design}
Suppose $m \geq \tilde{C}_4 nr \log n $ and let $A_1,\ldots,A_m$ be matrices of the form $a_j a_j^*$, where each $a_j \in \mathbb{C}^n$ is chosen independently from a weighted set $\left\{p_i,w_i \right\}_{i=1}^N$ of vectors obeying $\| w_i \|_{\ell_2}^2 = \sqrt{n(n+1)}$ for all $1 \leq i \leq N$ and 
\begin{equation}
\left\| \sum_{i=1}^N p_i w_i w_i^* - \sqrt{\frac{n+1}{n}} \id \right\|_\infty \leq \frac{1}{2}.
\label{eq:approx_tight_frame}
\end{equation}  
Then with probability at least $1 - \mathrm{e}^{-\gamma \tilde{C}_4 r}$, the matrix $W := \frac{1}{m}\sum_{j=1}^m A_j$ is positive definite and obeys
\begin{equation}
\max \left\{ \| W \|_\infty,\  \| W^{-1} \|_\infty,\ \kappa (W) \right\}  \leq 8.
\label{eq:W4design}
\end{equation}
Here, $\tilde{C}_4 >1$ and $0 < \gamma \leq 1$ denote absolute constants of adequate size.
\end{proposition}

Note that condition (\ref{eq:approx_tight_frame}) is slightly stronger than the corresponding condition  in Theorem \ref{Th2}.  Also, the construction of $W$ again uses $t = \frac{1}{m} \left( 1 \ldots,1 \right)^T \in \mathbb{R}^m$. 

\begin{proof}
In order to show this statement, we are going to employ the matrix Bernstein inequality\footnote{Resorting to the matrix Chernoff inequality would allow for establishing a similar result. However, in the case of an exact tight frame, the numerical constants obtained by doing so are slightly worse.} \cite[Theorem 6.1]{tropp_user_2012}, see also \cite{ahwi02},
in order to establish
\begin{equation}
\left\| W - \sqrt{\frac{n+1}{n}} \id \right\|_\infty \leq \frac{3}{4}
\label{eq:W4design_aux1}
\end{equation}
with high probability. Let $\lambda_1 (W),\ldots,\lambda_n (W)$ denote the eigenvalues of $W$. Then such a bound together with 
the definition of the operator norm  assures
\begin{align*}
1 - \lambda_{\min}(W) &\leq  \sqrt{\frac{n+1}{n}}- \lambda_{\min}(W)  \leq \left| \sqrt{\frac{n+1}{n}} - \lambda_{\min}(W) \right| \leq \max_{ 1 \leq i \leq n} \left| \sqrt{\frac{n+1}{n}}- \lambda_i (W) \right| \\&= \left\| \sqrt{\frac{n+1}{n}}\id - W \right\|_\infty \leq 3/4, \\
\lambda_{\max}(W) - \sqrt{\frac{n+1}{n}} & \leq \left| \lambda_{\max}(W) - \sqrt{\frac{n+1}{n}} \right| \leq \max_{1 \leq i \leq n} \left| \sqrt{\frac{n+1}{n}} - \lambda_i (W) \right|\\ & = \left\| W - \sqrt{\frac{n+1}{n}}\id \right\|_\infty \leq 3/4.
\end{align*}
This in turn implies
$
\lambda_{\min} (W) \geq 1/4 
$
as well as
$
\lambda_{\max}(W)\leq 3/4+\sqrt{\frac{n+1}{n}} \leq 2
$ for $n\geq 2$ and the desired bound \eqref{eq:W4design} readily follows.

It remains to assure the validity of \eqref{eq:W4design_aux1} with high probability.
To this end, for $1 \leq k \leq m$, we define the random matrices $M_k := \frac{1}{m} \left( a_k a_k^* - \mathbb{E} \left[ a_k a_k^* \right] \right)$, where each $a_k$ is chosen independently at random from the weighted set $\left\{ p_i,w_i \right\}_{i=1}^N$.
This definition assures
\begin{equation}
\left\| W - \sqrt{\frac{n+1}{n}}\id \right\|_\infty
= \left\| \sum_{k=1}^m \big(M_k + \mathbb{E} \left[ a_k a_k^* \right]\big) - \sqrt{\frac{n+1}{n}} \id \right\|_\infty
\leq \left\| \sum_{k=1}^m M_k \right\|_\infty + \frac{1}{2}
\label{eq:W4design_aux2}
\end{equation}
via the triangle inequality and assumption \eqref{eq:approx_tight_frame}
and along similar lines
\begin{equation}
\left\| \mathbb{E} \left[ a_k a_k^* \right] \right\|_\infty \leq \frac{1}{2} + \sqrt{\frac{n+1}{n}} \leq 2
\end{equation}
readily follows for any $1 \leq k \leq m$. 
The random matrices $M_k$ have mean-zero by construction and each of them obeys
\begin{align*}
\left\| M_k \right\|_\infty &= \frac{1}{m} \left\| a_k a_k^* - \mathbb{E} \left[ a_k a_k^* \right] \right\|_\infty
\leq \frac{1}{m} \max \left\{ \| a_k a_k^* \|_\infty, \| \mathbb{E} \left[ a_k a_k^* \right] \|_\infty \right\} 
= \frac{1}{m}\| a_k \|_{\ell_2}^2 = \frac{\sqrt{(n+1)n}}{m},
\end{align*}
as well as
\begin{align*}
\left\| \mathbb{E} \left[ M_k^2 \right] \right\|_\infty
&= \frac{1}{m^2} \left\| \mathbb{E} \left[ \left( a_k a_k^* \right)^2 \right] - \mathbb{E} \left[ a_k a_k^* \right]^2 \right\|_\infty 
= \frac{1}{m^2}\left\| \sqrt{(n+1)n} \mathbb{E } \left[ a_k a_k^* \right] - \mathbb{E} \left[ a_k a_k^* \right]^2 \right\|_\infty \\
&=\frac{2}{m^2} \max \left\{ \sqrt{(n+1)n} \left\| \mathbb{E} \left[  a_k a_k^* \right] \right\|_\infty, \left\| \mathbb{E} \left[ a_k a_k^* \right] \right\|_\infty^2 \right\} 
 \leq \frac{2 \sqrt{(n+1)n}}{m^2}. 
\end{align*}
Hence $$\left\|\sum_{k=1}^m \mathbb{E} \left[ M_k^2 \right] \right\|_\infty\leq  \frac{2 \sqrt{(n+1)n}}{m}.$$
These bounds allow us to set $R:= \frac{\sqrt{(n+1)n}}{m}$, $\sigma^2 := \frac{2 \sqrt{(n+1)n}}{m}$
and apply the matrix Bernstein inequality (\cite[Theorem 6.1]{tropp_user_2012}, \cite{ahwi02}) in order to establish
\begin{equation*}
\Pr \left[ \left\| \sum_{k=1}^m M_k \right\|_\infty \geq \tau \right] \leq  n\; \mathrm{exp}\left(-\frac{\tau^2/2}{\sigma^2+R\tau} \right) \leq n \; \mathrm{exp} \left( - \frac{3 \tau^2 m}{16 \sqrt{(n+1)n}} \right)
\end{equation*}
for $0 < \tau \leq \sigma^2/R =2$. 
Setting $\tau = 1/4$ and inserting $m \geq \tilde{C}_4 n r \log(n) $ (where $\tilde{C}_4 $ is large enough)
assures that \eqref{eq:W4design_aux1} holds with probability of failure smaller than $\mathrm{e}^{-\gamma \tilde{C}_4 r}$ via \eqref{eq:W4design_aux2} for a suitable $\gamma >0$.  
\end{proof}

Finally, we are ready to prove Theorem \ref{mainTh3}. 

\begin{proof}[Proof of Theorem \ref{mainTh3}]

We content ourselves with establishing the design case and point out that the Gaussian case can be proved analogously (albeit with different constants). 
Fix $0 < \rho < 1/8$ and suppose that $$m \geq C_3 \left( 1 + \left( 1 + \rho^{-1} \right)^2 \right) n r  \log n$$ measurement vectors have been chosen independently from an approximate 4-design.
Theorem \ref{Th2}  then assures that the resulting measurement operator $\mathcal{A}$ obeys the robust Frobenius rank null space property with constants $\rho < 1/8$ and $\tau \leq \tilde{C}_6/ \sqrt{m}$ with probability at least $1 - \mathrm{e}^{-\tilde{C}_5 m}$. 
Likewise, Proposition \ref{prop:W4design} assures that with probability at least $1 - \mathrm{e}^{-\gamma \tilde{C}_4 r}$, setting $t = \frac{1}{\sqrt{m}} (1,\ldots,1)^T \in \mathbb{R}^m$ 
leads to a positive definite $W = \sum_{j=1}^m t_j A_j$ obeying $\kappa (W) \leq 8$. Note that such a $t$ obeys $\| t \|_{\ell_2} = 1/\sqrt{m}$
and also $0 < \rho <1/8 \leq 1/ \kappa (W)$ holds by construction. 
The union bound over these two assertions failing implies that the requirements of Theorem \ref{th:NSPforPositiveMatrices} are met with probability at least
\begin{equation*}
1 - \mathrm{e}^{-\tilde{C}_5 m} - \mathrm{e}^{-\gamma \tilde{C}_4 r} \geq 1 - \mathrm{e}^{- \tilde{\gamma} \tilde{C}_4 r},
\end{equation*}
 where $\tilde{\gamma}$ denotes a sufficiently small absolute constant and $ \tilde{C}_4=m/n r  \log n$. The constants $C_4$ and $s$ presented in Theorem \ref{mainTh3}
then amount to $s=\tilde{\gamma} \tilde{C}_4$ and $C_2\geq \tilde{C}_4$.
Inserting $\| t \|_{\ell_2} = 1/\sqrt{m}$ and the bounds on $\| W \|_\infty, \| W^{-1} \|_\infty, \kappa (W)$ from Proposition \ref{prop:W4design} into \eqref{eq:NSPforPositiveMatrices}
yields
\begin{align*}
\norm{Z-X}_p & \leq\frac{2C\kappa(W)}{r^{1-1/p}}\norm{X_c}_1+r^{1/p-1/2}\norm{\nA(Z)-\nA(X)}_{\ell_2} \norm{W^{-1}}_{\infty} \brac{\frac{C\norm{t}_2}{\sqrt r}+D\norm{W}_{\infty}\tau} \\
&\leq \frac{16 C}{r^{1-1/p}} \norm{X_c}_1 + 8 r^{1/p-1/2} \| \mathcal{A} (Z) - \mathcal{A}(X) \|_{\ell_2} \left( \frac{C}{\sqrt{rm}} + \frac{9D \tilde{C}_6}{\sqrt{m}} \right) \\
& \leq \frac{C_3}{r^{1-1/p}} \norm{X_c}_1 + \frac{C_4 r^{1/p-1/2}}{\sqrt{m}} \left\| \mathcal{A} (Z) - \mathcal{A}(X) \right\|_{\ell_2}
\end{align*}
with constants $C_3 = 16 C$ and $C_4 = 8 C + 8 D \tilde{C}_6$ (where $C,D$ were introduced in Theorem \ref{th:NSPforPositiveMatrices} and $\tilde{C}_6$ is ). 
\end{proof}

\begin{remark} \label{rem:quantum_improvement}
In Corollary \ref{cor:tomography} we focus on recovering density operators, i.e., positive semidefinite matrices $X$ with trace one. 
This trace constraint can be re-interpreted as an additional perfectly noiseless measurement
\begin{equation*}
b_0 = \tr \left( \id X \right) = \tr (X) = 1
\end{equation*}
corresponding to the measurement matrix $A_0 = \id$. Setting $t = (1,0,\ldots,0)^T \in \mathbb{R}^{m+1}$ in Theorem \ref{th:NSPforPositiveMatrices}
then leads to $W = \id$ which obeys $\| W \|_{\infty} = \| W^{-1} \|_\infty = \kappa (W) = 1$ and furthermore assures that the endomorphism \eqref{eq:mapping} is trivial, i.e. $\tilde{Z} = Z$ for all $Z \in \mathcal{H}_n$.
Moreover, these properties render the estimate provided in Lemma \ref{Normdiff} redundant, because any two density operators $X,Z$ obey
\begin{equation*}
\| \tilde{Z} \|_1 - \| \tilde{X} \|_1
= \| Z \|_1 - \| Z \|_1 = \tr \left( Z \right) - \tr \left( X \right) =0.
\end{equation*}
Such a refinement then allows for dropping the term containing $\| t \|_{\ell_2}$ in \eqref{eq:NSPforPositiveMatrices} and by inserting $W = \id$ we arrive at the following conclusion: 
Any measurement operator $\mathcal{A}$ that obeys the Frobenius robust rank null space property with constants $0 < \rho <1 $ and $\tau >0$ 
assures for $1 \leq p \leq 2$ and any two density operators $X,Z$:
\begin{equation*}
\| Z - X\|_p \leq \frac{2 \left( 1 + \rho \right)^2}{1 - \rho} \norm{X_c}_1 + \tau \frac{r^{1/p - 1/2}(3 + \rho)}{1 - \rho} \| \mathcal{A}(Z) - \mathcal{A} (X) \|_{\ell_2}.
\end{equation*} 
Corollary \ref{cor:tomography} then follows from combining this assertion with Theorem \ref{Th2} and setting $p=1$.

\end{remark}

\subsection*{Acknowledgements}

MK, HR and UT acknowledge funding by the European Research Council through the Starting Grant StG 258926.
The work of RK is supported by the Excellence Initiative of the
German Federal and State Governments (Grants ZUK 43 \& 81), the ARO under
contracts W911NF-14-1-0098 and W911NF-14-1-0133 (Quantum
Characterization, Verification, and Validation), the Freiburg Research
Innovation Fund, the DFG (GRO 4334 \& SPP1798), and the 
State Graduate
Funding Program of Baden-W\"urttemberg.

\section*{Appendix}

\subsection*{A brief review of finite-dimensional quantum mechanics}

For the sake of being self-contained we briefly recapitulate crucial concepts of (finite dimensional) quantum mechanics without going too much into detail. For further reading on the topics introduced here, we defer the interested reader to \cite[Chapter 2.2]{nielsen_quantum_2010}. 

An isolated  quantum mechanical system is fully described by its \emph{density operator}.
For a finite $n$-dimensional quantum system, such a density operator corresponds to an Hermitian, positive semidefinite matrix $\rho$ with unit trace.

The most general notion of a measurement is that of a \emph{positive operator-valued measure} (POVM). 
For an $n$-dimensional quantum system, a  POVM corresponds to a collection $\mathcal{M}=\left\{ E_m \right\}_{m \in I}$ of positive semidefinite 
$n \times n$ matrices that sum up to identity, i.e.,
\begin{equation*}
\sum_{m \in I} E_m = \id.
\end{equation*}
The indices $m\in I$ indicate the possible measurement outcomes of performing such a POVM measurement.
Upon performing $\mathcal{M}$ on a system described by $\rho$, quantum mechanics then postulates that
the probability of obtaining the outcome (labeled by) $m$ corresponds to
\begin{equation*}
p (m, \rho) = \tr \left( E_m \rho \right).
\end{equation*}
Repeating the same measurement (i.e., preparing $\rho$ and measuring $\mathcal{M}$) many times allows one to estimate the $n$ probabilities $p (\lambda_i, \rho )$ ever more accurately. 

Note that the definitions of $\rho$ and $\mathcal{M}$ assure that ${p(m, \rho)}_{m \in I}$ is in fact a valid probability distribution. 
Indeed, $p(m,\rho) \geq 0$ follows from positive-semidefiniteness of both $\rho$ and $E_m$. 
Unit trace of $\rho$ assures proper normalization via
\begin{align*}
\sum_{m \in I} p (m, \rho ) = \sum_{m \in I} \tr \left( E_m \rho \right) = \tr \left( \id \rho \right) = \tr (\rho) = 1.
\end{align*}

\bibliographystyle{abbrv}

\end{document}